\documentclass{amsart}
\usepackage{amssymb}
\usepackage{amsmath}
\usepackage{amsthm}
\usepackage[arrow,matrix,curve]{xy}
\usepackage{graphicx}
\usepackage{hyperref}
\newcommand\datver[1]{\def\datverp%
   {\par\boxed{\boxed{\text{Run: \today}}}}}

\newcommand\boxb[1]{\square_b}

\newcommand\paperbody%
          {}

\newtheorem{theorem}{Theorem}[section]
\newtheorem{conjecture}[theorem]{Conjecture}

\newtheorem{lemma}[theorem]{Lemma}
\newtheorem{proposition}[theorem]{Proposition}
\newtheorem{corollary}[theorem]{Corollary}

\theoremstyle{remark}
\newtheorem{definition}[theorem]{Definition}
\newtheorem{remark}[theorem]{Remark}

\newtheorem{examples}[theorem]{Examples}
\newtheorem{example}[theorem]{Example}

\newcommand\co{\colon\,}

\newcommand\Tr{\operatorname{Tr}}

\newcommand\lp{\textup{(}}
\newcommand\rp{\textup{)}}

\newcommand\cA{\mathcal{A}}

\newcommand\cB{\mathcal{B}}
\newcommand\cE{\mathcal{E}}

\newcommand\cH{\mathcal{H}}
\newcommand\cK{\mathcal{K}}

\newcommand\cL{\mathcal{L}}

\newcommand\ep{e^\perp}

\newcommand\ZZ{\mathbb Z}

\newcommand\bbC{\mathbb C}

\newcommand\bbN{\mathbb N}
\newcommand\bbP{\mathbb P}
\newcommand\bbQ{\mathbb Q}
\newcommand\bbR{\mathbb R}

\newcommand\bbT{\mathbb T}
\newcommand\bbZ{\mathbb Z}
\newcommand\cD{\mathcal D}
\newcommand\bZ{\bbZ}
\newcommand\bR{\bbR}
\newcommand\bT{\bbT}
\newcommand\bC{\bbC}
\newcommand\bP{\bbP}

\newcommand\cC{\mathcal C}

\newcommand\rank{\operatorname{rank}}

\newcommand\Aut{\operatorname{Aut}}
\newcommand\End{\operatorname{End}}
\newcommand\Innbar{\overline{\operatorname{Inn}}}

\newcommand\Id{\operatorname{Id}}

\newcommand{\im}{\operatorname{im}}

\newcommand\Ca{$C^*$-algebra}

\datver{}
\begin{document}
\title{A noncommutative sigma-model}

\author{Varghese Mathai}
\address{Department of Pure Mathematics, University of Adelaide,
Adelaide, SA 5005, Australia}
\email{mathai.varghese@adelaide.edu.au}
\urladdr{http://www.maths.adelaide.edu.au/mathai.varghese/}
\author{Jonathan Rosenberg}
\thanks{Both authors thank the Erwin Schr\"odinger International Institute 
for Mathematical Physics for its hospitality and support under the
program in Gerbes, Groupoids, and Quantum Field Theory in 
Spring 2006, which made the beginning of this work possible. 
VM was partially supported by the Australian Research
Council. JR was partially supported by NSF Grants DMS-0504212
and DMS-0805003, and also
thanks the Department of  Pure Mathematics at The University of Adelaide
for its hospitality during visits in August 2007 and March 2009.}
\address{Department of Mathematics,
University of Maryland,
College Park, MD 20742, USA}
\email{jmr@math.umd.edu}
\urladdr{http://www.math.umd.edu/\raisebox{-.6ex}{\symbol{"7E}}jmr}
\begin{abstract}
We begin to study a sigma-model in which
both the spacetime manifold and the two-dimensional string world-sheet
are made noncommutative. We focus on the
case where both the spacetime manifold and the two-dimensional
string world-sheet are replaced by noncommutative $2$-tori.
In this situation, we are able to determine when maps between such
noncommutative tori exist, to derive the Euler-Lagrange equations,
to classify many of the critical
points of the Lagrangian, and to study the associated partition function.
\end{abstract}
\maketitle

\section{Introduction}
\label{sec:intro}
Noncommutative geometry is playing an increasingly important role in
physical field theories, especially quantum field theory and string
theory. Connes \cite{Co94} proposed a general formulation of action
functionals in noncommutative spacetime, and there is now a large
literature on noncommutative field theories (surveyed in part in
\cite{DN} and \cite{Sz03}). Thus it seems appropriate now to study fully
noncommutative sigma-models.

In our previous work \cite{MR,MR1,MR2}, we argued that
a consistent approach to T-duality for spacetimes $X$ which are principal
torus bundles over another space $Z$, with $X$
possibly equipped with a non-trivial H-flux, forces the
consideration of ``noncommutative'' T-duals in some situations.  A special
case of this phenomenon was also previously noted by Lowe, Nastase and 
Ramgoolam \cite{LNR}.

However, this work left open the question of what sort of sigma-model
should apply in the situation where the ``target space'' is no longer a space
at all but a noncommutative \Ca, and in particular (as this is the 
simplest interesting case), a noncommutative torus.

In classical sigma-models in string theory, the fields are maps
$g\co \Sigma \to X$, where $\Sigma$ is closed and $2$-dimensional, and
the target space $X$ is $10$-dimensional spacetime. The leading term
in the action is
\begin{equation}
\label{eq:classicalaction}
S(g)= \int_\Sigma \Vert \nabla g(x)\Vert^2 d\sigma(x),
\end{equation}
where the gradient and norm are computed with respect to suitable
Riemannian (or pseudo-Riemannian)
metrics on $\Sigma $ and $X$, $\sigma$ is volume measure
on $\Sigma$, and critical points of the action are just harmonic maps
$\Sigma\to X$. Usually one adds to \eqref{eq:classicalaction} a
Wess-Zumino term, related to the H-flux, an Einstein term, corresponding
to general relativity on $X$, and various other terms, but here we will
focus on \eqref{eq:classicalaction} (except in Section \ref{sec:WZ},
where the Wess-Zumino term will also come up).

The question we want to treat here is what should replace maps
$g\co \Sigma \to X$ and the action \eqref{eq:classicalaction}
when $X$ becomes noncommutative.  More precisely, we will
be interested in the case where we replace $C_0(X)$, the algebra
of continuous functions on $X$ vanishing at infinity, by a
noncommutative torus.  At the end of the paper, we will also comment
on what happens in the more complicated case, 
considered in \cite{MR}, \cite{MR1}, and \cite{MR2}, where
$A=\Gamma_0(Z,\cE)$ is the algebra of
sections vanishing at infinity of a continuous field $\cE$ of
noncommutative $2$-tori over a space $Z$, which plays the role of
reduced or ``physically observable'' spacetime.  (In other words, we
think of $X$ as a bundle over $Z$ with \emph{noncommutative} $2$-torus
fibers.)

Naively, since a map $g\co\Sigma \to X$ is equivalent to
a {\Ca} morphism $C_0(X)\to C(\Sigma)$, one's first guess would be to
consider $*$-homomorphisms $A\to C(\Sigma)$, where $\Sigma$ is still
an ordinary $2$-manifold. The problem with this approach when $A$ is 
complicated is that often there are
no such maps. For example, if $A = C_0(Z)\otimes A_\theta$ with
$\theta$ irrational (this is $\Gamma_0(Z,\cE)$ for a trivial field $\cE$ of
noncommutative tori over $Z$),
then simplicity of $A_\theta$ implies there are \emph{no} non-zero
$*$-homomorphisms $A\to C(\Sigma)$. Thus the first thing we see is that
once spacetime becomes noncommutative, it is necessary to allow 
the world-sheet $\Sigma$ to become noncommutative as well.

In most of this paper, we consider a sigma-model based on
$*$-homomorphisms between noncommutative $2$-tori. The first problem is to
determine when such maps exist, and this is studied in Section
\ref{sec:morphisms}. The main result here is Theorem
\ref{thm:homostoirrat}, which 
determines necessary and sufficient conditions for existence of a
non-zero $*$-homomorphism from $A_\Theta$ to $M_n(A_\theta)$, when
$\Theta$ and $\theta$ are irrational and $n\ge 1$. The main section of
the paper is Section \ref{sec:harmonic}, which studies an energy
functional on such $*$-homomorphisms. The critical points of the
energy are called harmonic maps, and we classify many of them when
$\Theta=\theta$. We also determine the Euler-Lagrange equations for
harmonic maps (Proposition \ref{prop:ELeq}), which are considerably
more complicated than in the commutative case. Subsection
\ref{sec:rational} deals in more detail with the special case of maps
from $C(\bT^2)$ to a rational noncommutative torus. Even this case is
remarkably complicated, and we discover interesting connections with
the field equations studied in \cite{DKL}. Section
\ref{sec:refinements} deals with various variations on the theory,
such as how to incorporate general metrics and the Wess-Zumino
term, and what happens when spacetime is a ``bundle'' of
noncommutative tori and not just a single noncommutative
torus. Finally, Section \ref{sec:physics} discusses what the partition 
function for our sigma-model may look like.

The authors are very grateful to Joachim Cuntz, Hanfeng Li, and the
referee of this paper for several helpful comments. They are
especially grateful to Hanfeng Li for writing the appendix \cite{HLi},
which resolves two problems which were unsolved when the first draft
of this paper was written.

\section{Classification of morphisms between irrational rotation algebras}  
\label{sec:morphisms}

In principle one should allow replacement of $\Sigma$ by general
noncommutative Riemann surfaces, as defined for example in \cite{Nat}
(in the case of genus $0$) and \cite{MP} (in the case of genus $>1$),
but since here we take our spacetimes to be noncommutative 
tori, it is natural to consider the ``genus one'' case and to 
replace $C(\Sigma)$ by $A_\theta$ for some $\theta$.
This case was already discussed and studied in
\cite{DKL}, but only in the case of exceptionally simple target spaces
$X$. In fact, in \cite{DKL0} and \cite{DKL}, $X$ was taken to be $S^0$, i.e.,
the algebra $A$ was taken to be $\bbC\oplus\bbC$. (Or alternatively,
one could say that they took $A=\bbC$, but allowed non-unital maps.)

We begin by classifying $*$-homomorphisms. We begin with the
(easy) case of unital maps.
\begin{theorem}
\label{thm:unitalhomostoirrat}
Fix $\Theta$ and $\theta$ in $(0,\,1)$, both irrational.
There is a unital $*$-homo\-mor\-phism $\varphi\co A_\Theta\to A_\theta$
if and only if $\Theta = c\theta+d$ for some $c,\,d\in\bbZ$, $c\ne
0$. Such a $*$-homomorphism $\varphi$ can be chosen to be
an isomorphism onto its image if and only if $c=\pm 1$.
\end{theorem}
\begin{proof}
Remember from \cite{Rieffel1,Rieffel2} that
projections in irrational rotation algebras are determined up to unitary
equivalence by their traces, that $K_0(A_\theta)$ is mapped
isomorphically to the ordered group $\bbZ + \theta\bbZ
\subset \bbR$ by the unique normalized trace 
$\Tr$ on $A_\theta$, and that the range of the trace $\Tr$ on
projections from $A_\theta$ itself is precisely $(\bbZ +
\theta\bbZ)\cap [0,\,1]$. 

Now a unital $*$-homomorphism $\varphi\co A_\Theta\to A_\theta$ must 
induce an order-preserving map $\varphi_*$ of $K_0$ groups sending the
class of the identity to the class of the identity. Since both $K_0$
groups are identified with dense subgroups of $\bbR$, with the
induced order and with the
class of the identity represented by the number $1$,
this map can be identified with the inclusion of a 
subgroup, with $1$ going to $1$. So $\Theta$, identified with a
generator of $K_0(A_\Theta)$, must lie in $\bbZ + \theta\bbZ$, say,
$\Theta = c\,\theta+d$ for some $c,\,d\in\bbZ$. That proves necessity of
the condition, but sufficiency is easy, since $A_{c\theta +d} \cong
A_{c\theta}$ is the universal {\Ca} on two unitaries $U$ and $V$
satisfying $UV = e^{2\pi i c \theta}VU$, while $A_\theta$ is the
universal {\Ca} on two unitaries $u$ and $v$ satisfying $uv=e^{2\pi i
  \theta}v u$. So define $\varphi$ by $\varphi(U) = u^c$, $\varphi(V)
= v$, and the required condition is satisfied. Note of course that if
$c=\pm 1$, then the images of $U$ and $V$ generate $A_\theta$ and
$\varphi$ is surjective, whereas if $|c|\ne 1$, then $\varphi_*$ is
not surjective (and so $\varphi$ can't be, either).
\end{proof}
\begin{remark}
With notation as in Theorem \ref{thm:unitalhomostoirrat}, if $c=\pm
1$, it is natural to ask if
it follows that \emph{any} $\varphi$ inducing the
isomorphism on $K_0$ is a $*$-isomorphism. The answer is definitely
``no.'' In fact, by \cite[Theorem 7.3]{EllAT}, which applies
because of \cite{ElEv}, for any given possible map $K_0(A_\Theta)\to
K_0(A_\theta)$, there is a $*$-homomorphism $A_\Theta\to A_\theta$
inducing any desired group homomorphism
$\bZ^2 \cong K_1(A_\Theta) \to K_1(A_\theta)\cong \bZ^2$, including
the $0$-map. In particular, $A_\theta$ always has proper (i.e.,
non-invertible) unital $*$-endomorphisms. (To prove this, take $\Theta
= \theta$, and observe that if the induced map on $K_1$ is not
invertible, then the endomorphism of $A_\theta$ cannot be invertible.)
It is not clear, however, whether or
not such endomorphisms constructed using the inductive limit structure
of \cite{ElEv} can be chosen to be smooth.

But Kodaka \cite{Kodaka,Kodaka1} has constructed smooth unital 
$*$-endomorphisms $\Phi$ of $A_\theta$, whose image has nontrivial
relative commutant, but only when $\theta$ is a quadratic
irrational of a certain type.  For a slight improvement on his result,
see Theorem \ref{thm:endos} below.

Note that the de la Harpe-Skandalis determinant $\Delta$ \cite{dlHS},
with the defining property
\[
\Delta(e^y) = \frac{\Tr(y)}{2\pi i}\mod \bbZ+\theta\bbZ,
\]
maps the abelianization of the connected component of the identity
in the unitary group of
$A_\theta$ to $\bbC^\times/(\bbZ+\theta\bbZ)$. Thomsen
\cite{Th} has proved that everything in the kernel of $\Delta$
is a \emph{finite product} of commutators. But for the element
$e^{2\pi i\theta}\in \ker \Delta$, we get a stronger result.
Since (by \cite{ElEv,EllAT}) $A_\theta$ has a proper $*$-endomorphism
$\varphi$ inducing the $0$-map on $K_1$, that means there are
two unitaries in $A_\theta$ (namely, $\varphi(U)$ and $\varphi(V)$)
in the connected component of the identity
in the unitary group with commutator $e^{2\pi i \theta}$.

As far as $*$-\emph{automorphisms} of $A_\theta$ are concerned, some
structural facts have been obtained by Elliott, Kodaka, and
Elliott-R\o rdam \cite{Elliott,Kodaka2,ER}. Elliott and R\o rdam \cite{ER}
showed that $\Innbar(A_\theta)$, the closure of the inner automorphisms,
is topologically simple, and that $\Aut(A_\theta)/\Innbar(A_\theta)
\cong GL(2,\bbZ)$. However, if one looks instead at \emph{smooth}
automorphisms, what one can call \emph{diffeomorphisms}, one sees
a different picture.  For $\theta$ satisfying a certain 
Diophantine condition \cite{Elliott}, $\Aut(A_\theta^\infty)$ is 
an iterated semidirect product, $(U(A_\theta^\infty)_0/\bbT) \rtimes
(\bbT^2\rtimes SL(2,\bbZ))$. This is not true without the Diophantine
condition \cite{Kodaka2}, but it may still be that
$\Aut(A_\theta^\infty) =\Innbar(A_\theta^\infty)\rtimes SL(2,\bbZ)$ 
for all $\theta$.  (See Elliott's review of
\cite{Kodaka2} in \emph{MathSciNet}.)
\end{remark}

Next we consider $*$-homomorphisms that are not necessarily unital. We
can attack the problem in two steps. If 
there is a non-zero $*$-homomorphism $\varphi\co A_\Theta\to 
M_\ell(A_\theta)$, not necessarily unital, then $\varphi(1_{A_\Theta})=p$ is a
self-adjoint projection, and $\im \varphi\subseteq p M_\ell(A_\theta)
p$, which is an algebra strongly Morita-equivalent to $A_\theta$. By
\cite[Corollary 2.6]{Rieffel2}, $p M_\ell(A_\theta) p$ must be
isomorphic to $M_n(A_\beta)$ for some $\beta$ in the orbit of $\theta$
under the action of $GL(2,\bbZ)$ on $\bbR$ by linear fractional
transformations. So we are essentially reduced to the unital case
covered in Theorem \ref{thm:unitalhomostoirrat}, except that we have
to allow for the possibility of passage to matrix algebras.  (This
would be the case even if $\ell=1$, since there is not necessarily any
relationship between $n$ and $\ell$.) This modification is covered in
the following:
\begin{theorem}
\label{thm:unitalhomostomatirrat}
Fix $\Theta$ and $\theta$ in $(0,\,1)$, both irrational, and $n\in
\bbN$, $n\ge 1$. There is a unital $*$-homo\-mor\-phism $\varphi\co
A_\Theta\to M_n(A_\theta)$ 
if and only if $n\Theta = c\theta+d$ for some $c,\,d\in\bbZ$, $c\ne
0$. Such a $*$-homomorphism $\varphi$ can be chosen to be
an isomorphism onto its image if and only if $n=1$ and $c=\pm 1$.
\end{theorem}
\begin{proof}
The argument is similar to that for Theorem
\ref{thm:unitalhomostoirrat}, since $K_0(M_n(A_\theta))$ is again
isomorphic (as an ordered group) to $\bbZ+\theta\bbZ$, but this time
the class of the identity is represented by $n$, so that if both $K_0$
groups are identified with subgroups of $\bbR$ in the usual way,
$\varphi_*$ must be multiplication by $n$. Hence if $\varphi$ exists,
$n\Theta \in \bbZ+\theta\bbZ$. 

For the other direction, suppose we
know that $n\Theta =c\theta+d$. We need to construct an embedding of
$A_\Theta$ into a matrix algebra over $A_\theta$. By \cite[Theorem
  4]{Rieffel1}, $A_\Theta= A_{(c\theta+d)/n}$ is strongly Morita equivalent to
$A_{n/(c\theta+d)}$, which embeds unitally into $A_{1/(c\theta+d)}$ as in
the proof of 
Theorem \ref{thm:unitalhomostoirrat}, and $A_{1/(c\theta+d)}$ is Morita
equivalent to $A_{c\theta+d}\cong A_{c\theta}$, which embeds
unitally into 
$A_\theta$. Stringing things together, we get an 
embedding of $A_\Theta$ into a matrix algebra over $A_{\theta}$. (By 
\cite[Proposition 2.1]{Rieffel1}, when two unital {\Ca}s are Morita
equivalent, each one embeds as a corner into a matrix algebra over the
other.) So we get a non-zero $*$-homomorphism $A_\Theta\to
M_\ell(A_\theta)$ (not necessarily unital), possibly with $\ell\ne
n$. The induced map 
$\varphi_*$ on $K_0$ can be identified with an order-preserving
homomorphism from $\bbZ + \left(\frac{c\theta+d}{n}\right)\bbZ$ to $\bbZ +
{\theta}\bbZ$. But in fact we can determine this map precisely, using
the fact \cite[page 425]{Rieffel1} that the Morita equivalence
from $A_{(c\theta+d)/n}$ to $A_{n/(c\theta+d)}$ is associated to multiplication
by $n/(c\theta+d)$, and the Morita equivalence from $A_{1/(c\theta+d)}$ to
$A_{c\theta+d}$ is associated to multiplication by $c\theta+d$. Thus the
composite map
$\varphi_*$ is multiplication by $n$, and sends the class of
$1_{A_\Theta}$ to $n$, which is the class of $1_n$ in
$K_0(M_\ell(A_\theta))$, where necessarily $\ell\ge n$. 
Since (by \cite{Rieffel2}) projections are
determined up to unitary equivalence by their classes in $K_0$, we can
conjugate by a unitary and arrange for $\varphi$ to map $A_\Theta$
unitally to $M_n(A_\theta)$.

For the last statement we use
\cite[Theorem 3]{Rieffel1}, which says that $A_\Theta$ can be isomorphic to
$M_n(A_\theta)$ only if $n=1$.
\end{proof}

We can now reorganize our conclusions in a way that is algebraically
more appealing. First, it's helpful in terms of motivation to point
out the following purely algebraic lemma, which we suspect is known,
though we don't know where to look it up.
\begin{lemma}
\label{lem:GL2Q}
Let $M$ be the submonoid {\lp}{\bfseries not} a subgroup{\rp} of $GL(2,\bbQ)$
consisting of matrices in $M_2(\bbZ)$ with non-zero determinant, i.e.,
of integral matrices having inverses that are not necessarily integral. Then
$M$ is generated by $GL(2,\bbZ)$ and by the matrices of the form
$\begin{pmatrix} r & 0 \\ 0 & 1\end{pmatrix}$, $r\in\bbZ\smallsetminus
  \{0\}$.
\end{lemma}
\begin{proof}
First we recall that applying an elementary row or column operation to
a matrix is the same as pre- or post-multiplying by an
elementary matrix of the
form $\begin{pmatrix} 1 & \star \\ 0 & 1\end{pmatrix}$ or
$\begin{pmatrix} 1 & 0 \\ \star & 1\end{pmatrix}$. So it will suffice
to show that, given any matrix $B\in M$, we can write it
as a product of matrices that reduce via elementary
row or column operations (over $\bbZ$) to things of the form
$\begin{pmatrix} \star & 0 \\ 0 & 1\end{pmatrix}$. The proof of this
is almost the same as for \cite[Theorem 2.3.2]{RosK}. Write $B =
\begin{pmatrix} b_{11} & b_{12} \\ b_{21} & b_{22}\end{pmatrix}$. Since $B$
is nonsingular, $b_{11}$ and $b_{21}$ can't both be $0$. Suppose
$b_{j1}$ is the smaller of the two in absolute value (or if the
absolute values are the same, choose $j=1$). Subtracting an integral
multiple of the $j$-th row from the other row, we can arrange to
decrease the minimal absolute value of the elements in the first
column. Proceeding this way and using the Euclidean algorithm, we can
reduce the first column to either $\begin{pmatrix} r\\ 0\end{pmatrix}$
or $\begin{pmatrix} 0\\ r\end{pmatrix}$ (with $r$ the greatest common
  divisor of the original $b_{11}$ and $b_{21}$). Since we can,
if necessary, left multiply by the elementary matrix
$\begin{pmatrix} 0 & 1 \\ -1& 0\end{pmatrix}$, we
can assume the first column has been reduced to $\begin{pmatrix} r\\
0\end{pmatrix}$, and thus that $B$ has been reduced to the form
\[
\begin{pmatrix} b_{11} & b_{12} \\ 0 & b_{22}\end{pmatrix}
= 
\begin{pmatrix} 1 & 0 \\ 0 & b_{22}\end{pmatrix}
\begin{pmatrix} 1 & b_{12} \\ 0 & 1\end{pmatrix}
\begin{pmatrix} b_{11} & 0 \\ 0 & 1\end{pmatrix}.
\]
And finally, $\begin{pmatrix} 1 & 0 \\ 0 & b_{22}\end{pmatrix}$ is
conjugate to $\begin{pmatrix} b_{22} & 0 \\ 0 & 1\end{pmatrix}$ under
  the elementary matrix $\begin{pmatrix} 0 & 1 \\ -1&
    0\end{pmatrix}$. 
\end{proof}
\begin{remark} 
\label{rem:Mgen}
We can relate this back to the proofs of Theorems
\ref{thm:unitalhomostoirrat} and \ref{thm:unitalhomostomatirrat}. 
Elements of $M$ lying in $GL(2,\bbZ)$ correspond to Morita equivalences
of irrational rotation algebras \cite[Theorem 4]{Rieffel1}.  Elements of
the form $\begin{pmatrix}
  r & 0 \\ 0 & 1\end{pmatrix}$ act on $\theta$ by multiplication by
  $r$, and correspond to inclusions $A_{r\theta} \hookrightarrow
  A_{\theta}$. Lemma \ref{lem:GL2Q} says that general elements of $M$
  are built out of these two cases. This motivates the following
  Theorem \ref{thm:homostoirrat}.
\end{remark}
\begin{remark} 
\label{rem:Hecke}
The appearance of the monoid $GL(2,\bbZ)\subset
  M\subset GL(2,\bbQ)$, and also the statement of Lemma
  \ref{lem:GL2Q}, are somewhat reminiscent of the theory of Hecke
  operators in the theory of modular forms, which also involve the
  action of the same monoid $M$ (on $GL(2,\bbR)/GL(2,\bbZ)$).
\end{remark}
\begin{theorem}
\label{thm:homostoirrat}
Fix $\Theta$ and $\theta$ in $(0,\,1)$, both irrational. 
Then there is a non-zero $*$-homomorphism $\varphi\co A_\Theta\to 
M_n(A_\theta)$ for some $n$, not necessarily unital, if and only if
$\Theta$ lies in the 
orbit of $\theta$ under the action of the monoid $M$ \textup{(of Lemma
\ref{lem:GL2Q})} on $\bbR$ by linear fractional transformations.
The possibilities for $\Tr(\varphi(1_{A_\Theta}))$ are precisely the
numbers $t = c\theta + d > 0$, $c,\,d\in \bbZ$ such that
$t\Theta\in \bbZ+ \theta\bbZ$. Once $t$ is chosen, $n$ can be taken to
be any integer $\ge t$.
\end{theorem}
\begin{proof}
First suppose $\varphi$ exists, and let $p
=\varphi(1_{A_\Theta})$. Then 
\[
\varphi_*\co K_0(A_\Theta) \to 
K_0(M_n(A_\theta))=K_0(A_\theta)
\]
must be an injection of ordered
groups sending $1\in K_0(A_\Theta)$ to $t=\Tr(p)=c\theta + d\in
\bbZ+\theta\bbZ$. Since both groups are dense subgroups 
of $\bbR$, this map must be multiplication by $t$ and must send
$\Theta$ to something in $\bbZ + \theta\bbZ$. So we have
$t\Theta = a\theta + b$ for some $a,\,b\in\bbZ$, and 
\[
\Theta = \frac{a\theta + b}{c\theta + d}= 
\begin{pmatrix}a&b\\c&d\end{pmatrix}\cdot \theta.
\]
The matrix $\begin{pmatrix}a&b\\c&d\end{pmatrix}$ has integer entries,
  and can't be singular since the numerator and denominator are both
  non-zero (being $\Tr(\varphi(q))$ and $\Tr(\varphi(1))$,
  respectively, where $q$ is a Rieffel projection in $A_\Theta$ with
  trace $\Theta$), and $(a\ b)$ and $(c\ d)$ can't be rational multiples of
  each other (as that would imply $\Theta$ is rational).
Hence $\begin{pmatrix}a&b\\c&d\end{pmatrix}$ lies in $M$, and $t$ is
  as required. And since $p\le 1_n$, $t\le n$.

To prove the converse, suppose
$A=\begin{pmatrix}a&b\\c&d\end{pmatrix}\in M$ and
$(c\theta + d)\Theta = a\theta + b$. Let $t= c\theta + d$ and choose
any integer $n\ge t$. Since the range of the trace on projections in
$M_n(A_\theta)$ is $[0,\,n]\cap (\bbZ + \theta\bbZ)$, we can choose a
self-adjoint projection $p\in M_n(A_\theta)$ with $\Tr(p)=t$. The
subalgebra $p M_n(A_\theta) p$ of $M_n(A_\theta)$ is a full corner
(since $A_\theta$ is simple), hence is strongly Morita equivalent to
$A_\theta$, hence is $*$-isomorphic to $M_k(A_\beta)$ for some $\beta$
in the orbit of $GL(2,\bbZ)$ acting on $\theta$ \cite[Corollary
  2.6]{Rieffel2}. In fact, we can
compute $k$ and $\beta$; $k$ is the (positive) greatest common divisor
of $c$ and $d$, and $\beta$ is obtained by completing the row vector
$\bigl(\frac{c}{k}\ \frac{d}{k}\bigr)$ to a matrix 
\[
\begin{pmatrix}a'&b'\\ \frac{c}{k} & \frac{d}{k}\end{pmatrix} \in
GL(2,\bbZ) 
\]
and then letting this act on $\theta$. By Theorem
\ref{thm:unitalhomostomatirrat}, there is a $*$-homomorphism
$\varphi\co A_\Theta\to p M_n(A_\theta) p \cong M_k(A_\beta)$ with
$\varphi(1_{A_\Theta}) = p$ if and only if
$k\Theta\in\bbZ+\beta\bbZ$. But, by assumption,
\[
k\Theta = k \frac{a\theta + b}{c\theta + d}
= \frac{a\theta + b}{\frac{c}{k}\theta +\frac{d}{k}}
\]
while
\[
\beta = \frac{a'\theta + b'}{\frac{c}{k}\theta +\frac{d}{k}}
\quad \text{with} \quad
\begin{pmatrix}a'&b'\\ \frac{c}{k} & \frac{d}{k}\end{pmatrix} \in
GL(2,\bbZ).
\]

Note that the transpose matrix
\[
\begin{pmatrix}a'& \frac{c}{k} \\ b' & \frac{d}{k}\end{pmatrix}
\]
also lies in  $GL(2,\bbZ)$.
So we can we can solve for integers $r$ and $s$ such that
\[
\begin{pmatrix}a'& \frac{c}{k} \\ b' & \frac{d}{k}\end{pmatrix}
\begin{pmatrix}r \\ s\end{pmatrix} =
\begin{pmatrix}a \\ b\end{pmatrix}.
\]
That says exactly that
\[
\begin{aligned}
r\beta + s &= r\cdot \frac{a'\theta + b'}{\frac{c}{k}\theta
  +\frac{d}{k}} + s\\
&= \frac{r\bigl(a'\theta + b'\bigr) + s\bigl(\frac{c}{k}\theta
  +\frac{d}{k}\bigr)}{\frac{c}{k}\theta
  +\frac{d}{k}} \\
&= \frac{\bigl(a'r+\frac{c}{k}s\bigr)\theta +
\bigl(b'r+\frac{d}{k}s\bigr)}{\frac{c}{k}\theta
  +\frac{d}{k}} \\
&= \frac{a\theta + b}{\frac{c}{k}\theta
  +\frac{d}{k}}\\
&= k\Theta,
\end{aligned}
\]
as required.
\end{proof}

\section{Harmonic maps between noncommutative tori} 
\label{sec:harmonic}

\subsection{The action and some of its minima
for maps between noncommutative tori} 
\label{sec:action-tori}
In this section we consider the analogue of the action functional
\eqref{eq:classicalaction} in the context of the $*$-homomorphisms
classified in the last section. For simplicity, consider first of all
a unital $*$-homomorphism $\varphi\co A_\Theta \to A_\theta$ as in
Theorem \ref{thm:unitalhomostoirrat}. As before, denote the canonical
generators of $A_\Theta$ and $A_\theta$ by
$U$ and $V$, $u$ and $v$, respectively. The natural analogue of 
$S(g)$ in our situation is
\begin{equation}
\label{eq:qlag}
\begin{aligned}
\cL(\varphi) &= \Tr\Bigl(\delta_1(\varphi(U))^*\delta_1(\varphi(U)) +
\delta_2(\varphi(U))^*\delta_2(\varphi(U)) \\ &\qquad +
\delta_1(\varphi(V))^*\delta_1(\varphi(V)) +
\delta_2(\varphi(V))^*\delta_2(\varphi(V))\Bigr).
\end{aligned}
\end{equation}
(Except for a factor of two, this is the same as the sum of the 
``energies'' of the unitaries $\varphi(U)$ and $\varphi(V)$ in
$A_\theta$, as defined in \cite[\S5]{RosL}.) 
Here $\delta_1$ and $\delta_2$ are the infinitesimal generators for
the ``gauge action'' of the group $\bbT^2$ on $A_\theta$. More
precisely, $\delta_1$ and $\delta_2$ are defined on the smooth
subalgebra $A_\theta^\infty$ by the formulas
\begin{equation}
\label{eq:ders}
\delta_1(u) = 2\pi iu, \quad \delta_2(u)=0, \quad 
\delta_1(v) = 0, \quad \delta_2(v)=2\pi iv.
\end{equation}
The derivations $\delta_1$ and $\delta_2$ play the role of measuring partial
derivatives in the two coordinate directions in 
$A_\theta$ (which, we recall, plays the role of the worldsheet
$\Sigma$), the product of an operator with its
adjoint has replaced the norm squared, 
and integration over $\Sigma$ has been replaced by the
trace.
Note for example that if $\Theta=\theta$ and $\varphi=\Id$, the
identity map, then  we obtain
\[
\cL(\Id) = \Tr\Bigl(\delta_1(u)^*\delta_1(u) + 0 + 0 +
\delta_2(v)^*\delta_2(v)\Bigr) = 8\pi^2.
\]
More generally, for the $*$-automorphism $\varphi_A\co u\mapsto u^pv^q$,
$v\mapsto u^rv^s$, with $A=\begin{pmatrix}p&q\\r&s\end{pmatrix}\in
  SL(2,\bbZ)$, we obtain
\begin{equation}
\label{eq:minenergy}
\begin{aligned}
\cL( \varphi_A) &= \Tr\Bigl(\delta_1(u^pv^q)^*\delta_1(u^pv^q)
+ \delta_2(u^pv^q)^*\delta_2(u^pv^q) \\
&\quad + \delta_1(u^rv^s)^*\delta_1(u^rv^s)
+ \delta_2(u^rv^s)^*\delta_2(u^rv^s)\Bigr)\\
&= 4\pi^2\Bigl(p^2 + q^2 + r^2 + s^2\Bigr).
\end{aligned}
\end{equation}
\begin{conjecture}
\label{conj:minenergy}
The value \textup{\eqref{eq:minenergy}} of $\cL(\varphi_A)$ is
minimal among all $\cL(\varphi)$, $\varphi\co
A_\theta^\infty\circlearrowleft$
a $*$-endomorphism inducing the matrix $A\in SL(2,\bbZ)$ on
$K_1(A_\theta) \cong \bbZ^2$.
\end{conjecture}
Note that this conjecture is a close relative of \cite[Conjecture
  5.4]{RosL}, which deals with maps $C(S^1)\to A_\theta$ instead of
maps $A_\Theta\to A_\theta$. That conjecture said that the multiples
of $u^mv^n$ minimize the energy of the unitaries in their connected
components. Since $\cL(\varphi)$ is twice the sum of the energies of
$\varphi(U)$ and $\varphi(V)$, \cite[Conjecture
  5.4]{RosL} immediately implies the present conjecture.
The following results provide support for Conjecture \ref{conj:minenergy}.
\begin{theorem}
Conjecture \textup{\ref{conj:minenergy}} is true if $\varphi\co
A_\theta^\infty\circlearrowleft$ maps $u$ to a scalar multiple
of itself. {\lp}In
this case, $p=s=1$ and $q=0$.{\rp} The minimum is achieved precisely
when $\varphi(v) = \lambda u^r v$, $\lambda\in \bbT$.
\end{theorem}
\begin{proof}
Let $\varphi(u) = \mu u$ and 
$\varphi(v)=w$, where $\mu \in \bbT$, $w$ is unitary and
smooth, and (necessarily) $uw=e^{2\pi i\theta}wu$. Since also $uv =
e^{2\pi i\theta} vu$, it follows that $wv^*$ is a unitary commuting with
$u$. Since the $*$-subalgebra generated by $u$ is maximal abelian,
that implies that $w = f(u)v$, where $f\co \bbT\to \bbT$ is
continuous, and the parameter $r$ is the winding number of $f$. Now we
compute that $\delta_1(f(u)v) = 2\pi i f'(u)uv$, $\delta_2(f(u)v) =
2\pi i f(u) v$, and hence:
\begin{equation}
\label{eq:enofphi}
\begin{aligned}
\cL(\varphi) &= \Tr\Bigl(\delta_1(u)^*\delta_1(u) +
\delta_2(u)^*\delta_2(u) \\ &\qquad +
\delta_1(f(u)v)^*\delta_1(f(u)v) +
\delta_2(f(u)v)^*\delta_2(f(u)v)\Bigr)\\
&= 4\pi^2\Tr\Bigl(2 + v^*u^*f'(u)^*f'(u)uv\Bigr)\\
&= 4\pi^2\Tr\Bigl(2 + f'(u)^*f'(u)\Bigr).
\end{aligned}
\end{equation}
We can pull $f\co \bbT\to \bbT$ back to a function $[0,\, 1] \to
\bbR$ via the covering map $z=e^{2\pi i t}$, and then the
winding number of $f$ (as a self-map of $\bbT$) translates into the
difference $f(1)-f(0)$ (for $f$ defined on $[0,\, 1]$).
The problem of minimizing \eqref{eq:enofphi} is thus the same as that
of minimizing $\int_0^1\vert f'(t)\vert^2 \,dt$ in the class of smooth
functions $f\co [0,\,1] \to \bbR$ with $f(1)-f(0)=r$. Since such a
function can be written as $f(t) = f(0) + tr + g(t)$, with
$g(0)=g(1)=0$, and $f'(t) = r + g'(t)$, we have
\[
\int_0^1\vert f'(t)\vert^2 \,dt =
\int_0^1 \bigl( r^2 + 2 r g'(t) + g'(t)^2 \bigr) \,dt =
r^2 + \Vert g'\Vert^2_{L^2} \ge r^2,
\]
with equality exactly when $g'\equiv 0$, i.e., $g$ constant, and thus
$g\equiv 0$ since $g(0)=0$. Thus equality occurs when (going back to
the original notation) $f(u)=\lambda u^r$, i.e., $\varphi(v) = \lambda
u^r v$, for some constant $\lambda\in\bbT$.
\end{proof}

We now give a complete proof of Conjecture \ref{conj:minenergy} 
for $*$-automorphisms,
in the case where the Diophantine condition of \cite{Elliott} is
satisfied. The same proof works in general modulo a technical
point which we will discuss below.
\begin{theorem}
\label{thm:varhomphi}
Conjecture \textup{\ref{conj:minenergy}} is true for $*$-automorphisms,
assuming the Diophantine condition of \textup{\cite{Elliott}}
is satisfied. In other words, if $\varphi$ is an $*$-automorphism
of $A_\theta^\infty$ inducing the map given by $A\in SL(2,\bbZ)$
on $K_1(A_\theta)$, and if $\theta$ satisfies the Diophantine 
condition of \textup{\cite{Elliott}}, then
\[
\cL(\varphi)\ge \cL(\varphi_A),
\]
with equality if and only if $\varphi(u)=\lambda \varphi_A(u)$,
$\varphi(v)=\mu \varphi_A(v)$, for some $\lambda,\mu\in\bbT$. 
\end{theorem}
\begin{proof}
What we use from \cite{Elliott} is that the hypothesis on $\theta$
ensures that we can write
$\varphi(u)=\lambda w \varphi_A(u) w^*$,
$\varphi(v)=\mu w\varphi_A(v) w^*$, for some $\lambda,\mu\in\bbT$
and for some unitary $w\in A_\theta^\infty$. Suppose
$A = \begin{pmatrix}p&q\\r&s\end{pmatrix}\in
SL(2,\bbZ)$. Since $\cL(\varphi)$ is a sum of four terms, all of
which have basically the same form, it will be enough to estimate
the first term; the estimate for the other three is precisely analogous.
We find that
\begin{multline*}
\delta_1(\varphi(u))=\delta_1(\lambda w u^p v^q w^*) \\
= \lambda\bigl(\delta_1(w)  u^p v^q w^* + w 2 \pi i p  
u^p v^q w^* + w u^p v^q \delta_1(w)^* \bigr),
\end{multline*}
so the first term in $\cL(\varphi)$, $\Tr\bigl(\delta_1(\varphi(u))^*
\delta_1(\varphi(u)) \bigr)$ is a sum of nine terms, three
``principal'' terms and six ``cross'' terms.  Note that
$\overline\lambda$ in $\delta_1(\varphi(u))^*$ cancels the
$\lambda$ in $\delta_1(\varphi(u))$, so we can ignore the $\lambda$
altogether. The three principal terms are
\begin{equation}
\label{eq:prin}
\begin{aligned}
&\Tr\Bigl(
\bigl(\delta_1(w)  u^p v^q w^* \bigr)^*
\bigl(\delta_1(w)  u^p v^q w^* \bigr) \\ & \qquad +
\bigl( w  2 \pi i p u^p v^q w^*\bigr)^*
\bigl( w  2 \pi i p u^p v^q w^*\bigr) \\ & \qquad +
\bigl( w u^p v^q \delta_1(w)^* \bigr)^*\bigl( w u^p v^q \delta_1(w)^* \bigr)
\Bigr)\\ &=
2 \Tr\bigl( \delta_1(w)^* \delta_1(w) \bigr) + 4\pi^2 p^2,
\end{aligned}
\end{equation}
where in the last step we have used (several times)
the invariance of the trace under inner automorphisms.

Now consider the six cross-terms. These are
\begin{equation}
\label{eq:cross}
\begin{aligned}
&\Tr\Bigl(
\bigl(\delta_1(w)  u^p v^q w^* \bigr)^*
\bigl( w  2 \pi i p \,u^p v^q w^*\bigr) \\ & \qquad +
\bigl(\delta_1(w)  u^p v^q w^* \bigr)^*
\bigl( w u^p v^q \delta_1(w)^* \bigr)
\\ & \qquad +
\bigl( w  2 \pi i p\, u^p v^q w^*\bigr)^*
\bigl(\delta_1(w)  u^p v^q w^* \bigr)
\\ & \qquad +
\bigl( w  2 \pi i p \, u^p v^q w^*\bigr)^*
\bigl( w u^p v^q \delta_1(w)^* \bigr)
\\ & \qquad +
\bigl( w u^p v^q \delta_1(w)^* \bigr)^*
\bigl(\delta_1(w)  u^p v^q w^* \bigr)
\\ & \qquad +
\bigl( w u^p v^q \delta_1(w)^* \bigr)^*
\bigl( w  2 \pi i p\, u^p v^q w^*\bigr) \Bigr) \\ & =
\Tr\Bigl(
2 \pi i p \,\delta_1(w)^* w + w \bigl( u^p v^q \bigr)^*
\delta_1(w)^*  w u^p v^q \delta_1(w)^*  \\ & \qquad 
- 2 \pi i p \, w^*\delta_1(w) - 2 \pi i p \, w \, \delta_1(w)^* \\ & \qquad
+ \delta_1(w)\bigl( u^p v^q \bigr)^* w^* \delta_1(w)  u^p v^q w^*
+  2 \pi i p \, \delta_1(w) w^* \Bigr)\\ 
&=  \Tr\Bigl(
w \bigl( u^p v^q \bigr)^*
\delta_1(w)^*  w u^p v^q \delta_1(w)^*  \\ & \qquad 
+ \delta_1(w)\bigl( u^p v^q \bigr)^* w^* \delta_1(w)  u^p v^q w^*
\Bigr).
\end{aligned}
\end{equation}
(Note the use of ``integration by parts,'' \cite[Lemma 2.1]{RosL}.)
Now we put \eqref{eq:prin} and \eqref{eq:cross} together. We obtain
\[
\begin{aligned}
\Tr\bigl(\delta_1(\varphi(u))^*
\delta_1(\varphi(u)) \bigr) &= 4\pi^2 p^2 +
\Tr \Bigl( 2\delta_1(w)^* \delta_1(w)\\ & \qquad \qquad 
+ w \bigl( u^p v^q \bigr)^*
\delta_1(w)^*  w u^p v^q \delta_1(w)^* \\ & \qquad \qquad 
+ \delta_1(w)\bigl( u^p v^q \bigr)^* w^* \delta_1(w)  u^p v^q w^*
\Bigr).
\end{aligned}
\]
We make the substitutions $T=\delta_1(w)^* w$ and $W=u^pv^q$. Note
that $W$ is unitary. We obtain
\[
\begin{aligned}
\Tr\bigl(\delta_1(\varphi(u))^*
\delta_1(\varphi(u)) \bigr) &= 4\pi^2 p^2 +
\Tr \Bigl( T T^* + T^* T + W^* T  W T + T^* W^* T^* W\Bigr) \\
&\text{(using invariance of the trace under cyclic permutations)}\\
&= 4\pi^2 p^2 + \Tr \Bigl( T^* W W^* T + T^* W T^* W^*
+ W T W^* T \\
& \qquad\qquad \qquad + W T T^* W^*\Bigr)\\
&= 4\pi^2 p^2 + \Tr \Bigl( \bigl( W^* T + T^* W^* \bigr)^*
\bigl( W^* T + T^* W^* \bigr) \Bigr)\\
& \ge 4\pi^2 p^2.
\end{aligned}
\]
Furthermore, equality holds only if $W^* T + T^* W^* = 0$, i.e.,
$T = - W T^* W^*$.
Similar estimates with the other three terms in the energy show that
$\cL(\varphi) \ge \cL(\varphi_A) = 4\pi^2 \bigl( p^2 + q^2 + r^2 + s^2
\bigr)$, with equality only if $\delta_j(w)^* w = -  W w^* \delta_j(w) W^*$
and $\delta_j(w)^* w = -  W_1 w^* \delta_j(w) W_1^*$, where $W_1
= u^r v^s$.  (The conditions involving $W_1$ come from the analysis
of the last two terms in $\cL(\varphi)$, which use the
\emph{second} row of the matrix $A$.)  So if equality holds,
$W$ and $W_1$ both conjugate $w^* \delta_j(w)$ to 
the negative of its adjoint.  In
particular, $w^* \delta_j(w)$ commutes with $W^*W_1$. But this
unitary generates a maximal abelian subalgebra, so $w^* \delta_j(w)$ 
is a function $f$ of $W^*W_1$. So $w^* \delta_j(w) = f(W^*W_1)$ with
$W f(W^* W_1) W^* = - f(W^*W_1)^*$.  One can check that these equations
can be satisfied only if $f=0$. Indeed, we have the commutation
relation $W W_1 = e^{2\pi i \theta}W_1 W$, so
\[
W (W^* W_1)^n W^* = (W_1 W^*)^n = e^{2\pi i n \theta} (W^* W_1)^n.
\]
If we expand $f$ in a Fourier series, $f(W^* W_1) = \sum_n 
c_n (W^* W_1)^n $, then we must have
\[
\begin{aligned}
- f(W^*W_1)^* &= - \sum_n \overline{c_n} (W^* W_1)^{-n}
= - \sum_n \overline{c_{-n}} (W^* W_1)^{n} \\
& = \sum_n c_n
W (W^* W_1)^n W^*  = \sum_n c_n e^{2\pi i n \theta} 
(W^* W_1)^n.
\end{aligned}
\]
Equating coefficients gives 
\[
- \overline{c_{-n}} =  c_n e^{2\pi i n \theta}, \quad
\text{and replacing $n$ by $-n$,}\quad
- \overline{c_n} =  c_{-n} e^{-2\pi i n \theta}.
\]
These give 
\[
-c_{-n} = \overline{c_n}  e^{- 2\pi i n \theta}
= -  c_{-n} e^{-4\pi i n \theta},
\]
so all $c_n$ must vanish for $n\ne 0$. Thus $f$ is a constant
equal to its negative, i.e., $f=0$,
so $\delta_1(w)=0$ and $\delta_2(w)=0$,
$w$ is a scalar, and $\varphi$ differs from $\varphi_A$ only by a 
gauge transformation. That completes the proof.
\end{proof}
\begin{remark}
\label{rem:nonDio}
Note that the same proof always shows that $\cL(\varphi_A)
\le \cL(\varphi)$ for any $\varphi$ in the orbit of $\varphi_A$
under gauge automorphisms and inner automorphisms, and thus, by
continuity, under automorphisms in the closure 
(in the topology of pointwise $C^\infty$ convergence) of the inner
automorphisms.  So if the conjecture of Elliott that
$\Aut(A_\theta^\infty) =\Innbar(A_\theta^\infty)\rtimes SL(2,\bbZ)$ 
mentioned earlier
is true, the Diophantine condition in Theorem \ref{thm:varhomphi} is
unnecessary.
\end{remark}
\begin{remark}
After the first draft of this paper was written, Hanfeng Li succeeded
in proving \cite[Conjecture  5.4]{RosL} and Conjecture
\ref{conj:minenergy}  (in complete generality). His solution is given
in the appendix \cite{HLi}.
\end{remark}
\begin{remark}
\label{rem:enofendo}
Of course, so far we have neglected smooth proper $*$-endomorphisms of
$A_\theta$, which by \cite{Kodaka,Kodaka1} certainly exist at least
for certain quadratic irrational values of $\theta$. We do not know if
one can construct such endomorphisms to be energy-minimizing.
But we can slightly improve the result of \cite{Kodaka1} as follows.
\end{remark}
\begin{theorem}
\label{thm:endos}
Suppose $\theta$ is irrational.  Then there is a {\lp}necessarily
injective{\rp} unital
$*$-endomorphism $\Phi\co A_\theta \to A_\theta$, with image
$B \subsetneq A_\theta$ having
non-trivial relative commutant and with a conditional expectation of
index-finite type from
$A_\theta$ onto $B$, if and only if $\theta$ is a quadratic
irrational number. When this is the case, $\Phi$ can be chosen to be smooth.
\end{theorem}
\begin{proof}
The ``only if'' direction and the idea behind the ``if'' direction are
both in \cite{Kodaka1}. We just need to modify his construction as
follows.  Suppose $\theta$ is a quadratic irrational. Thus there exist
$a,\,b,\,c\in\bZ$ with $a\theta^2+b\theta+c=0$, $a\ne0$. Choose $d\in
\bZ$ with $0 < a\theta + d < 1$, and let $e$ be an orthogonal
projection in $A_\theta$ with trace $a\theta + d$. Since
\[
(a\theta + d)\theta = a\theta^2 + d\theta = (d-b)\theta - c \in
\bZ+\theta\bZ, 
\]
by Theorem \ref{thm:homostoirrat}, there is an injective $*$-homomorphism
$\varphi_1\co A_\theta \to A_\theta$ with image $eA_\theta e$.
Let $\ep = 1-e$.  Since $\Tr(1-e) = -a\theta + 1 - d$ and
\[
(-a\theta + 1 - d)\theta = -a\theta^2 + (1-d)\theta = (1+b-d)\theta + c \in
\bZ+\theta\bZ, 
\]
there is also an injective $*$-homomorphism
$\varphi_2\co A_\theta \to A_\theta$ with image $\ep A_\theta
\ep$. Since $eA_\theta e$ and $\ep A_\theta\ep$ are orthogonal, $\Phi
= \varphi_1+ \varphi_2$ is a unital 
$*$-endomorphism of $A_\theta$ whose image has $e$ in its relative
commutant. It is clear (since $e$ can be chosen smooth) that $\Phi$
can be chosen to be smooth.
The last part of the argument can be taken more-or-less
verbatim from \cite{Kodaka}. Let
\[
\Psi(x)= \frac12\Bigl( exe + \ep x\ep + \varphi_2(\varphi_1^{-1}(exe))
+ \varphi_1(\varphi_2^{-1}(\ep x\ep)) \Bigr).
\]
Then $\Psi$ is a faithful conditional expectation onto the image of
$\Phi$, and it has index-finite type as shown in \cite[\S2]{Kodaka}.
\end{proof}
\begin{remark}
\label{rem:end0}
As pointed out earlier by Kodaka, the endomorphisms constructed in
Theorem \ref{thm:endos} can be constructed to implement a wide
variety of maps on $K_1$. In fact, one can even choose $\Phi$ so that
$\Phi_*=0$ on $K_1$, with $\Phi$ taking both $u$ and $v$ to the
connected component of the identity in the unitary group. One can see
this as follows. The map $\Phi$ constructed in
Theorem \ref{thm:endos} can be written as $\iota\circ \Delta$, where
$\Delta\co A_\theta \to A_\theta \times A_\theta$ is the diagonal map
and $\iota$ is an inclusion of $A_\theta \times A_\theta$ into
$A_\theta $ (which exists for $\theta$ a quadratic irrational). Since
``block direct sum'' agrees with the addition in $K_1$, it follows
that (in the notation of the proof above) $\Phi_* 
= (\varphi_1)_*+(\varphi_2)_*$ on $K_1$. One can easily arrange to have $
(\varphi_1)_* = (\varphi_2)_* = \Id$, which would make $\Phi_* =
\text{multiplication by }2$. But if $\varphi_3$ is
the automorphism of $A_\theta$ with
$u\mapsto u^{-1}$, $v\mapsto v^{-1}$ and we replace $\Phi = \iota\circ
\Delta$ by $\Phi'=\iota\circ (\Id\times \varphi_3)\circ \Delta$, then since
$(\varphi_3)_*=-1$ on $K_1$, we get an
endomorphism $\Phi'$ inducing the $0$-map on $K_1$.

In fact, one can modify the construction so that $\Phi_*$ is any
desired endomorphism of $K_1$. So far  we have seen how to get
$\Phi_*=2$ or $\Phi_*=0$. To get $\Phi_*=1$, use a construction with
three blocks. In other words, choose mutually orthogonal projections
$e$ and $f$ in $A_\theta$ so that there exist $*$-isomorphisms 
$\varphi_1$, $\varphi_2$, and $\varphi_3$ from
$A_\theta$ onto each of $eA_\theta e$, $fA_\theta f$, and
$(1-e-f)A_\theta (1-e-f)$, respectively. (With $a$, $b$, $c$, $d$ as
above, this can be done by choosing $\Tr e = (a\theta+d)^2$
and $\Tr f = (a\theta+d)(1-d-a\theta)$.) As above, one can arrange to have
$(\varphi_1)_*=(\varphi_2)_*=1$ on $K_1$ and $(\varphi_3)_*=-1$. So if
$\Phi = \varphi_1 +\varphi_2 + \varphi_3$, $\Phi$ is a unital
$*$-endomorphism inducing multiplication by $1+1-1=1$ on $K_1$.
Other cases can be done similarly.
\end{remark}

\subsection{Euler-Lagrange equations}
\label{sec:ELeq}
In Proposition \ref{prop:ELeq} below, we determine the Euler-Lagrange
equations for the  
energy functional, $\cL(\varphi)$ in \eqref{eq:qlag}. One striking
difference with the classical  commutative case,
is that one cannot get rid of the ``integral'' $\Tr$ in the
Euler-Lagrange equations whenever $\theta$ is irrational. In Corollary
\ref{cor:satisfiesEL}, we construct  explicit harmonic maps with
respect to $\cL$. 

\begin{proposition}\label{prop:ELeq}
Let $\cL(\varphi)$ denote the energy functional for a unital
$*$-endo\-mor\-phism $\varphi$
of $A_\theta$. Then the Euler-Lagrange equations for $\varphi$ to be a
{\bfseries harmonic map}, that is, a critical point of $\cL$, are:
\begin{equation*}
0  = \sum_{j=1}^2\Big\{\Tr\left(A\,\delta_j\left[\varphi(u)^*
\delta_j(\varphi(u))\right]\right) 
+ \Tr\left(B\,\delta_j\left[\varphi(v)^*\delta_j
(\varphi(v)) \right]\right) 
\Big\}
\end{equation*}
where $A, B$ are self-adjoint elements in $A_\theta$,
constrained to satisfy the equation,
$$A-\varphi(v)^*A\varphi(v) = B-\varphi(u)^*B\varphi(u).$$
\end{proposition}

\begin{proof}
Consider the 1-parameter family of $*$-endomorphisms of $A_\theta$
defined by
\begin{align*}
\varphi_t(u) &= \varphi(u) e^{ih_1(t)} \\
& = \varphi(u)[ 1 + ith_1'(0) + O(t^2)],\\
\varphi_t(v) &= \varphi(v) e^{ih_2(t)} \\
& = \varphi(u)[ 1 + ith_2'(0) + O(t^2)],
\end{align*}
where $h_j(t), \, j=1, 2$ are 1-parameter families of self-adjoint
operators with
$h_1(0)=0=h_2(0)$.
Therefore
$$
\delta_j(\varphi_t(u)) = \delta_j(\varphi(u))  +  it \delta_j(\varphi 
(u))
h_1'(0) +  it \varphi(u) \delta_j( h_1'(0)) + O(t^2),
$$
and taking adjoints,
\[
\delta_j(\varphi_t(u))^* = \delta_j(\varphi(u))^*  - it h_1'(0) \delta_j
(\varphi(u))^* - it  \delta_j( h_1'(0))\varphi(u)^* + O(t^2),
\]
and similarly with $v$ in place of $u$, $h_2$ in place of $h_1$.
Using this, the term of order $t$ in $\Tr\left(\delta_j(\varphi_t(u))^*
\delta_j(\varphi_t(u)) \right)$
equals
\begin{equation}\label{eqn:der1}
\begin{aligned}
& i\Tr\left(\delta_j( h_1'(0))\left( \delta_j(\varphi(u))^*\varphi(u) -  
\varphi
(u)^* \delta_j(\varphi(u))\right)\right)\\
&= -2i \Tr\left(\delta_j( h_1'(0))\varphi(u)^* 
\delta_j(\varphi(u))\right).
\end{aligned}
\end{equation}
(Here we used the fact that since $\varphi(u)$ is unitary,
$\delta_j(\varphi(u))^*\varphi(u) + \varphi(u)^*\delta_j(\varphi(u))=0$.)
Because of ``integration by parts'' \cite[Lemma 2.1]{RosL},
equation \eqref{eqn:der1} equals
\begin{equation}\label{eqn:der2}
2i\Tr\left(h_1'(0)\,
\delta_j\left[\varphi(u)^*\delta_j(\varphi(u))\right]\right).
\end{equation}
Similarly, we calculate the term of order $t$ in $\Tr\left(\delta_j
(\varphi_t(v))^* \delta_j(\varphi_t(v)) \right)$
to be
\begin{equation}\label{eqn:der3}
2i\Tr\left(h_2'(0)\,
\delta_j\left[\varphi(v)^*\delta_j(\varphi(v))\right]\right).
\end{equation}
Setting $A=h_1'(0), \, B=h_2'(0)$, we deduce that the Euler-Lagrange
equations for $\cL$, defined by
$0=\frac{d}{dt}\cL(\varphi_t)\Big|_{t=0}$, are given as in the  
Proposition.

We next differentiate the constraint equations,
\begin{align*}
0&= \frac{d}{dt}\left(\varphi_t(u)\varphi_t(v) - e^{2\pi i\theta} 
\varphi_t(v)
\varphi_t(u)\right)\Big|_{t=0}   \\
& = \varphi(u)h_1'(0)\varphi(v) +  \varphi(u)\varphi(v) h_2'(0) - e^{2 
\pi i\theta}
\left[\varphi(v)h_2'(0)\varphi(u)
+ \varphi(v)\varphi(u)h_1'(0)\right].
\end{align*}

Using the fact that $\varphi$ is a unital $*$-endomorphism of $A_\theta$,
that is, $\varphi$ satisfies
$$
\varphi(u)\varphi(v) = e^{2\pi i\theta}\varphi(v)\varphi(u)
$$
we easily see that the constraint equations of the Proposition are
also valid.
\end{proof}
The following is not especially interesting since it is already
implied by the stronger result in \cite{HLi}, but it illustrates
how one might check this condition in some cases.
\begin{corollary}
\label{cor:satisfiesEL}
If $\varphi_A$  is the $*$-automorphism
of $A_\theta^\infty$ defined by
$\varphi_A(u) = u^p v^q$ and
$\varphi_A(v) = u^r v^s$, with $A=\begin{pmatrix}p&q\\r&s\end{pmatrix} 
\in  SL(2,\bbZ)$, then $\varphi_A$ is a critical point of $\cL(\varphi)$.
\end{corollary}

\begin{proof}
We compute:
\begin{align*}
\delta_1(\varphi_A(u)) & = 2\pi i p \varphi_A(u), \quad \delta_2 
(\varphi_A(u))  = 2\pi i q\varphi_A(u), \\
\delta_1(\varphi_A(v)) &= 2\pi i r \varphi_A(v), \quad \delta_2 
(\varphi_A(v)) = 2\pi i s \varphi_A(v).
\end{align*}
Therefore
\begin{align*}
\varphi_A(u)^*\delta_1(\varphi_A(u)) & =  2\pi i p, \\
\varphi_A(u)^*\delta_2(\varphi_A(u)) & =  2\pi i q, \\
\varphi_A(v)^*\delta_1(\varphi_A(v)) & =  2\pi i r, \\
\varphi_A(v)^*\delta_2(\varphi_A(v)) & =  2\pi i s, 
\end{align*}
Applying any derivation $\delta_j, \, j=1,2$, to any of the terms  
above gives zero, since they are all constants. Therefore
$\varphi_A$ is a critical point of $\cL$, by the Euler-Lagrange
equations in Proposition \ref{prop:ELeq}. 
\end{proof}
Of course, a major question is to determine how many critical points
there are for $\cL$ \emph{aside from} those of the special form
$\varphi_A$, $A\in SL(2,\bZ)$.

\subsection{Certain maps between rational noncommutative tori}
\label{sec:rational}
In this subsection we investigate certain harmonic maps between
\emph{rational} noncommutative tori. This is an exception to our
general focus on irrational rotation algebras, but it might shed some
light on what seems to be the most difficult case, of 
(possibly nonunital) maps
$\varphi\co A_\Theta \to M_m(A_\theta)$ implementing a Morita equivalence
when 
\[
\Theta = 1/\theta = \begin{pmatrix}0&1\\1&0
\end{pmatrix}\cdot \theta. 
\]
In effect, we consider this same situation,
but in the case where $\Theta = n>1$ is a positive integer, so that
$A_\Theta=C(\bT^2)$, the universal $C^*$-algebra generated by two
commuting unitaries $U$ and $V$. In this case, $A_\theta=A_{1/n}$ is an
algebra of sections of a bundle over $\bT^2$ with fibers
$M_n(\bC)$. This bundle is in fact the endomorphism bundle of a
complex vector bundle $V$ over $\bT^2$, with Chern class $c_1(V)\equiv
1$ (mod $n$). (More generally, $A_{k/n}$ is the algebra of sections of
the endomorphism bundle of a vector bundle of Chern class $\equiv
k$ (mod $n$); one can see this, for instance, from the explicit
description of the algebra in \cite{Brab}.) If $u$ and $v$ are the canonical
unitary generators of 
$A_{1/n}$, then $u^n$ and $v^n$ are both central, and generate the
center of $A_{1/n}$, which is isomorphic to $C(T^2)$, the copy of
$\bT^2$ here being identified with the spectrum of the algebra $A_{1/n}$. Since
the normalized trace on $A_{1/n}$ sends $1$ to $1$, it takes the value
$\frac{1}{n}$ on rank-one projections $e$, which exist in abundance. 
(The fact that there are lots of global rank-one projections is due to
the fact that the Dixmier-Douady invariant of the algebra vanishes.)
A choice of $e$ determines a $*$-isomorphism $\varphi_e$ from 
$A_0=A_n=C(\bT^2)$ to $eA_{1/n}e$, sending $U$ to $eu^n$ and $V$ to
$ev^n$. Let us compute the action functional on $\varphi_e$.
\begin{proposition} 
\label{prop:ratLandi}
With notation as above, i.e., with $e$ a
  self-adjoint projection in $A_{1/n}$ and
\[
\varphi_e\co C(\bT^2)\xrightarrow{\cong} e A_{1/n}e,\qquad
\varphi_e(U)=eu^n,\,\varphi_e(V)=ev^n,
\]
we have
\[
\cL(\varphi_e) = 2 \Tr \left(\delta_1(e)^2 + \delta_2(e)^2 +
4\pi^2n^2\right)\,. 
\]
Thus, up to a renormalization, this is the same as the action
functional on $e$ as defined in \cite{DKL0,DKL}. Thus $\varphi_e$ is
harmonic exactly when $e$ is harmonic.
\end{proposition}
\begin{proof}
We have
\[
\delta_1(eu^n)=\delta_1(e)\,u^n + 2\pi i n\, e\,u^n =
\bigl(\delta_1(e)+2\pi i n\, e\bigr)u^n \text{ and } 
\delta_2(eu^n)=\delta_2(e)\,u^n\,,
\]
and similarly for $ev^n$ (with the roles of $\delta_1$ and $\delta_2$
reversed). Since $u^n$ and $v^n$ are central, they cancel out when we
compute $\bigl(\delta_1(eu^n)\bigr)^*\delta_1(eu^n)$, etc., and we
obtain
\[
\begin{aligned}
\bigl(\delta_1(eu^n)\bigr)^*\delta_1(eu^n) &= \bigl(\delta_1(e)+2\pi i
n\,e\bigr)^* \bigl(\delta_1(e)+2\pi i n\,e\bigr) \\
& = \bigl(\delta_1(e)\bigr)^2 + 2\pi i n
\bigl(\delta_1(e)e-e\delta_1(e)\bigr) +4\pi^2 n^2,\\
\bigl(\delta_2(eu^n)\bigr)^*\delta_2(eu^n) &=
\bigl(\delta_2(e)\bigr)^2,\\
\bigl(\delta_1(ev^n)\bigr)^*\delta_1(ev^n) &=
\bigl(\delta_1(e)\bigr)^2,\\
\bigl(\delta_2(ev^n)\bigr)^*\delta_2(ev^n) &=
\bigl(\delta_2(e)\bigr)^2 + 2\pi i n
\bigl(\delta_2(e)e-e\delta_2(e)\bigr) +4\pi^2 n^2,\\
\end{aligned}
\]
and the result follows since the ``cross-terms'' have vanishing trace.
\end{proof}
While a complete classification seems difficult, we at least have an
existence theorem.
\begin{theorem}
\label{thm:harmonicrat}
There exist harmonic nonunital $*$-isomorphisms $\varphi_e\co C(\bT^2)
\to A_{1/n}$.
\end{theorem}
\begin{proof}
By Proposition \ref{prop:ratLandi}, it suffices to show that $A_{1/n}$
contains harmonic rank-$1$ projections. In terms of the realization of
$A_{1/n}$ as $\Gamma(T^2, \End(V))$, the sections of the endomorphism
bundle of the complex vector bundle $V$, this is equivalent to showing
that $\bP(V)$, the $\bC\bP^{n-1}$-bundle over $\bT^2$ whose fiber at a
point $x$ is the projective space of $1$-dimensional subspaces of $V_x$,
has harmonic sections for its natural connection.

One way to prove this is by using holomorphic geometry. Realize $\bT^2$
as an elliptic curve $E=\bC/(\bZ+i\bZ)$ and $V$ as a holomorphic
bundle. Then a holomorphic section of $\bP(V)$ is certainly
harmonic. But a holomorphic section of $\bP(V)$ will exist provided $V$ has an
everywhere non-vanishing  holomorphic section $s$, since the line
through $s(z)$ is a point of $\bP(V_z)$ varying holomorphically with
$z$. Since $n = \rank V > \dim E = 1$, this is possible by \cite[Theorem
  2, p.\ 426]{AtVB}, assuming that $V$ has ``sufficient holomorphic
sections,'' i.e., that there is a holomorphic section through any
point in any fiber. The condition of having sufficient sections is
weaker than being ample, which we can arrange by changing $c_1(V)$ to
be sufficiently positive (recall that only $c_1(V)$ mod $n$ is fixed,
so we have this flexibility).
\end{proof}
In preparation for Example \ref{ex:rat2} below, it will be useful to
give a concrete model for the algebra $A_{1/n}$.
\begin{proposition}
\label{prop:A1n}
Let $n>1$, and let $\zeta = e^{2\pi i/n}$. Fix the $n\times n$
matrices
\[
u_0 = \begin{pmatrix} 0 & 1 & 0& \cdots\\ 0 & 0 & 1 &\cdots\\
0 & 0 & 0 &\ddots\\
1& 0 & 0 &\cdots \end{pmatrix},\qquad
v_0 = \begin{pmatrix} 1 & 0 & 0& \cdots\\ 0 & \zeta & 0  &\cdots\\
0 & 0 & \zeta^2 &\cdots\\
0& 0 & 0 &\ddots \end{pmatrix},
\]
or in other words $v_0 =
\text{\textup{diag}}(1,\zeta,\zeta^2,\cdots,\zeta^{n-1})$.
Then $A_{1/n}$ can be identified with the algebra of continuous
functions $f\co \bT^2 \to M_n(\bC)$ satisfying the transformation rules
\[
\begin{cases} 
f(\zeta\lambda,\,\mu) = v_0^{-1} f(\lambda,\,\mu) v_0,\\
f(\lambda,\,\zeta\mu) = u_0 f(\lambda,\,\mu) u_0^{-1}.
\end {cases}
\]
\end{proposition}
\begin{proof}
Observe that $u_0^n=v_0^n=1$ and that $u_0 v_0 = \zeta v_0 u_0$. It is
then easy to see that the most general irreducible representation of
$A_{1/n}$ is equivalent to one
of the form $\pi_{\mu,\lambda}\co u\mapsto \mu u_0,\, v\mapsto \lambda v_0$
for some $(\mu,\,\lambda) \in \bT^2$.  However, we are
``overcounting,'' because it is clear that $v_0^{-1}$ conjugates
$\pi_{\mu,\lambda}$ to $\pi_{\zeta\mu,\lambda}$, and $u_0$ conjugates
$\pi_{\mu,\lambda}$ to $\pi_{\mu,\zeta\lambda}$. The spectrum of the
algebra $A_{1/n}$ can thus be identified with the quotient of $\bT^2$
by the action by multiplication by $n$-th roots of unity in both
coordinates. The result easily follows.
\end{proof}
\begin{example}
\label{ex:rat2}
We now give a specific example of this situation in which one can
write down an explicit harmonic map. We suspect one can do something
similar in general, but to make the calculations easier, we 
restrict to the case $n=2$. Proposition \ref{prop:A1n} describes
$A_{1/2}$ as the algebra of continuous functions $f\co \bT^2\to
M_2(\bC)$ satisfying 
\begin{equation}
\label{eq:fconds}
\begin{aligned}
f(-\lambda,\mu) &= \begin{pmatrix} 1 & 0\\ 0 & -1\end{pmatrix}
  f(\lambda,\mu) \begin{pmatrix} 1 & 0\\ 0 & -1\end{pmatrix},\\
f(\lambda,-\mu) &= \begin{pmatrix} 0 & 1\\ 1 & 0\end{pmatrix}
  f(\lambda,\mu) \begin{pmatrix} 0 & 1\\ 1 & 0\end{pmatrix}.
\end{aligned}
\end{equation}
If we write $\lambda=e^{i\theta_1}$ and $\mu=e^{i\theta_2}$, we
can rewrite \eqref{eq:fconds} by thinking of 
\[
f=\begin{pmatrix}
f_{11}&f_{12}\\f_{21}&f_{22}
\end{pmatrix}
\]
as defined on $[0,\,\pi]\times [0,\,\pi]$, subject
to boundary conditions
\begin{equation}
\label{eq:fboundconds}
\begin{alignedat}{2}
f_{11}(\pi,\theta_2)&=f_{11}(0,\theta_2), &\qquad
f_{22}(\pi,\theta_2)&=f_{22}(0,\theta_2), \\
f_{12}(\pi,\theta_2)&=-f_{12}(0,\theta_2), &\qquad
f_{21}(\pi,\theta_2)&=-f_{21}(0,\theta_2), \\
f_{11}(\theta_1,\pi)&=f_{22}(\theta_1,0), &\qquad
f_{22}(\theta_1,\pi)&=f_{11}(\theta_1,0), \\
f_{12}(\theta_1,\pi)&=f_{21}(\theta_1,0), &\qquad
f_{21}(\theta_1,\pi)&=f_{12}(\theta_1,0).
\end{alignedat}
\end{equation}
To get a nonunital harmonic map inducing an isomorphism from $C(\bT)$
to a nonunital subalgebra of $A_{1/2}$, we need by Proposition
\ref{prop:ratLandi} to choose $f$ satisfying \eqref{eq:fboundconds} so that
for all $\theta_1$ and $\theta_2$, $f(\theta_1,\theta_2)$ is self-adjoint with
trace $1$ and determinant $0$, and so that $f$ is harmonic. The
conditions \eqref{eq:fboundconds} as well as the conditions for $f$ to be a
rank-one projection will be satisfied provided that $f$ is of the form:
\begin{equation}
\label{eq:f}
f\left({\theta_1},\, {\theta_2}\right) 
= \frac{1}{2}\begin{pmatrix} 1 +\cos(g(\theta_1))\,\cos\theta_2
& \sin(g(\theta_1)) - i \cos(g(\theta_1))\,\sin\theta_2\\
\sin(g(\theta_1)) + i \cos(g(\theta_1))\,\sin\theta_2  &
1-\cos(g(\theta_1))\,\cos\theta_2
\end{pmatrix}
\end{equation}
with $g$ real-valued and satisfying the conditions
\begin{equation}
\label{eq:gh}
g(0)=-\frac{\pi}{2},\qquad g(\pi)=\frac{\pi}{2}.
\end{equation}
For $f$ to be harmonic, we need to make sure it satisfies the
Euler-Lagrange equation $f(\Delta f) = (\Delta f)f$, which is derived
in \cite[\S4.1]{DKL}. In the realization of Proposition
\ref{prop:A1n}, the canonical generators of $A_{1/2}$ are given by
\[
u\left(e^{i\theta_1},\, e^{i\theta_2}\right) =  e^{i\theta_1}u_0,\qquad
v\left(e^{i\theta_1},\, e^{i\theta_2}\right)= e^{i\theta_2}v_0, 
\]
so that $\delta_1$ and
$\delta_2$ act by $2\pi\frac{\partial}{\partial \theta_1}$ and 
$2\pi\frac{\partial}{\partial
\theta_2}$, respectively. Thus up to a factor of $4\pi^2$, $\Delta$
can be identified with the usual Laplacian in the variables $\theta_1$
and $\theta_2$. A messy calculation, which we performed with
$\text{\textsl{Mathematica}}^{\text{\textregistered}}$, though one can
check it by 
hand, shows that the commutator of $f$ and $\Delta f$ vanishes exactly
when the function $g$ in \eqref{eq:f} satisfies the nonlinear (pendulum)
differential equation 
\[
2g''(\theta) + \sin(2 g(\theta))=0.
\]
Subject to the boundary conditions \eqref{eq:gh}, this has a unique
solution, which \textsl{Mathematica} plots as in Figure \ref{fig:gtheta}.
\begin{figure}[htb]
\label{fig:gtheta}
\begin{center}
\includegraphics[height=2in]{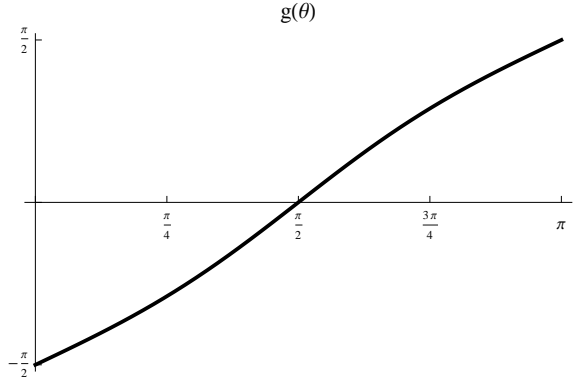}
\caption{Plot of $g(\theta)$ as computed by \textsl{Mathematica}}
\end{center}
\end{figure}

Note incidentally that \textsl{Mathematica} calculations show that
this solution is neither self-dual nor anti-self-dual, in the sense of
\cite{DKL}. In fact, writing out the self-duality and
anti-self-duality equations for a projection of the form \eqref{eq:f}
shows that they reduce to $g'(\theta)=\pm\cos(g(\theta))$, so
the only self-dual or anti-self-dual projections of this
form satisfying the initial condition $g(0)=-\pi/2$ are constant (and
thus don't satisfy the other boundary condition in \eqref{eq:gh}).

It may be of interest to compute the value of $\cL$ for this
example. The normalized trace $\Tr_A$ on $A_{1/2}$ for matrix-valued functions
$f$ satisfying \eqref{eq:fboundconds} is
\[
\Tr_A f = \frac{1}{2\pi^2}\int_0^\pi\int_0^\pi
\Tr f(\theta_1,\theta_2)\,d\theta_1\, d\theta_2,
\]
so by Proposition \ref{prop:ratLandi},
\[
\begin{aligned}
\cL(\varphi_f) & = 2 \Tr_A \left(\delta_1(f)^2 + \delta_2(f)^2 +
4\pi^2n^2\right)\text{ with } n=2\\
&= 2 \Tr_A \left(4\pi^2 \left(\frac{\partial f}{\partial
  \theta_1}\right)^2+ 4\pi^2 \left(\frac{\partial f}{\partial
  \theta_2}\right)^2+ 4\pi^2\cdot 4\right) \\
&=  8\pi^2 \left(4 + \frac{1}{2\pi^2}\int_0^\pi\int_0^\pi \Tr\left(
\left(\frac{\partial f}{\partial   \theta_1}\right)^2
+ \left(\frac{\partial f}{\partial   \theta_2}\right)^2
\right) \,d\theta_1\, d\theta_2\right) \\
&=  8\pi^2 \left(4 + 0.1116\right) \approx 32.89 \pi^2.
\end{aligned}
\]
(The integral was computed numerically with \textsl{Mathematica}.)
\end{example}

\section{Variations and Refinements}
\label{sec:refinements}
One can argue that what we have done up till now was somewhat special,
in that we took a very special form for the metric on the
``worldsheet,'' and ignored the Wess-Zumino term in the action.
In this section, we discuss how to generalize the results given
earlier in the paper. The modifications to the proofs given in the  
earlier sections are routine, and most arguments will not be repeated.

\subsection{Spectral triples and sigma-models}
\label{sec:spectriple}
In this subsection, we write a general sigma-model energy functional for
spectral triples, that specializes to the cases existing in the
literature, including what was discussed earlier in the paper.
It is an explicit variant of the discussion in \cite[\S VI.3]{Co94}
and \cite[\S2]{DKL}. Recall 
that a spectral triple $(\cA, \cH, D)$ is given by an
involutive unital algebra $\cA$ represented as bounded operators
on a Hilbert space $\cH$ and a self-adjoint operator $D$ with compact
resolvent such that the
commutators $[D,a]$ are bounded for all $a\in \cA$.
A spectral triple $(\cA, \cH, D)$ is said to be \emph{even} if the
Hilbert space $\cH$
is endowed with a $\ZZ_2$-grading $\gamma$ which commutes
with all $a\in \cA$ and anti-commutes with $D$.
Suppose in addition that  $(\cA, \cH, D)$ is $(2, \infty)$-summable,
which means (assuming for simplicity that $D$ has
no nullspace) that $\Tr_\omega(a |D|^{-2}) <\infty$, where $\Tr_\omega$ denotes the Dixmier trace.
We recall from VI.3 in \cite{Co94} that
$$
\psi_2(a_0, a_1, a_2) = \Tr((1+\gamma) a_0 [D, a_1][D, a_2])
$$
defines a positive Hochschild $2$-cocycle on $\cA$, where $\gamma =
\left(\begin{array}{cc}1 & 0 \\0 & -1\end{array}\right)$ is the grading
operator on $\cH$, and where $\Tr$ denotes the Dixmier trace composed with
$D^{-2}$.  In this paper,  although we consider the
canonical trace $\Tr$ instead of the above trace, all the properties go
through with either choice.
Using  the Dixmier trace $\Tr_\omega$ composed with $D^{-2}$
  has the advantage of {\em scale
invariance}, i.e., it is invariant under the replacement of $D$ by $ 
\lambda D$
for any nonzero $\lambda \in \bbC$, which becomes relevant when one  
varies the
metric, although for special classes of metrics, the scale invariance  
can be obtained by other means also.
The positivity of $\psi_2$ means that $\langle a_0 \otimes a_1, b_0  
\otimes b_1 \rangle
= \psi_2(b_0^*a_0, a_1, b_1^*)$ defines a positive sesquilinear form on
$\cA \otimes \cA$.

We now give a prescription for energy functionals in the sigma-model
consisting of homomorphisms $\varphi\colon \cB \longrightarrow \cA$,
from a smooth subalgebra of a $C^*$-algebra $\cB$ with target the  
given even $(2, \infty)$-summable spectral triple
$(\cA, \cH, D)$. Observing that $\varphi^*(\psi_2)$ is a 
positive Hochschild $2$-cocycle on $\cB$, we need to choose a
formal ``metric'' on $\cB$, which is a positive element $G\in  
\Omega^2(\cB)$ in the space of universal $2$-forms on $\cB$. Then evaluation
$$
\cL_{G, D}(\varphi) = \varphi^*(\psi_2)(G) \ge 0
$$
defines a general sigma-model action.

Summarizing, the data for a general sigma-model action consists of
\begin{enumerate}
\item A $(2, \infty)$-summable spectral triple
$(\cA, \cH, D)$;
\item A positive element $G\in \Omega^2(\cB)$
in the space of universal $2$-forms on $\cB$, known as a metric on $\cB$.
\end{enumerate}

Consider a unital $C^*$-algebra
generated by the $n$ unitaries $\{U_j: i=1, \ldots n\}$, with
finitely many relations as in \cite{Loring}, and let $\cB$ be a
suitable subalgebra consisting of rapidly vanishing series whose terms
are (noncommutative) monomials in the $U_i$'s. Then a choice of positive
element $G \in \Omega^2(\cB)$ (or metric on $\cB$) is given by
$$
G = \sum_{j,k=1}^n G_{jk}(dU_j)^*dU_k,
$$
where the matrix $(G_{jk})$ is symmetric, real-valued, and positive
definite. Then we compute the energy functional in this case,
$$
\cL_{G,D}(\varphi) = \varphi^*(\psi_2)(G)  =
\sum_{j,k=1}^n G_{jk}\Tr((1+\gamma) [D, \varphi(U_j)^*][D,  
\varphi(U_k)]) \ge 0.
$$

The \emph{Euler-Lagrange equations} for $\varphi$ to be a
critical point of $\cL_D$ can be derived as in Proposition
\ref{prop:ELeq}, but since the equations are long, we omit them.

We next give several examples of this sigma-model energy functional.
In all of these cases, the target algebra $\cA$ will be $A_\theta^ 
\infty$.
The first example is the
D\c abrowski-Krajewski-Landi model \cite{DKL}, consisting of
non-unital $*$-homomorphisms
$\varphi\colon \bbC \longrightarrow A_\theta^\infty$. Note that
$\varphi(1)=e $ is a projection in
the noncommutative torus $A_\theta$, and for any
$(2, \infty)$-summable spectral triple  $(A_\theta^\infty, \cH, D)$
on the noncommutative torus, our sigma-model energy functional is
$$
\cL_D(\varphi) = \Tr\left[(1+\gamma) [D,e][D,e]\right].
$$
Choosing the even spectral triple given by $\cH = L^2(A_\theta)
\otimes \bbC^2$ consisting of
the Hilbert space closure of $A_\theta$ in the canonical scalar
product coming from the trace,
tensored with the $2$-dimensional representation space
of spinors, and $D =\gamma_1 \delta_1 + \gamma_2 \delta_2$, where
$$
\gamma_1 = \left(\begin{array}{cc}0 & 1 \\1 & 0\end{array}\right) ,
\quad
\gamma_2 = \left(\begin{array}{cc}0 & -i  \\ i &
0\end{array}\right)
$$
are the Pauli matrices, we calculate that
$$
\cL_D(\varphi) = \sum_{j=1}^2\Tr\left[(\delta_j e)^2\right],
$$
recovering the action in \cite{DKL} and the Euler-Lagrange equation
$(\Delta e)e = e(\Delta e)$ there.

Next, we consider the model in Rosenberg \cite[\S5]{RosL}, consisting of
unital $*$-homo\-mor\-phisms
$\varphi\colon C(S^1) \longrightarrow A_\theta^\infty$. Let $U$ be the
unitary given by multiplication by
the coordinate function $z$ on $S^1$ (considered as the unit circle
$\bT$ in $ \bbC$). Then  $\varphi(U)$ is a unitary in
the noncommutative torus $A_\theta$, and for any $(2, \infty)$-summable
spectral triple  $(A_
\theta^\infty, \cH, D)$
on the noncommutative torus, our sigma-model energy functional is
$$
\cL_D(\varphi) = \Tr\left[(1+\gamma)[D, \varphi(U)^*][D,  
\varphi(U)]\right].
$$
Choosing the particular spectral triple on the noncommutative torus as
above, we calculate that
   $$
\cL_D(\varphi) = \sum_{j=1}^2\Tr\left[(\delta_j (\varphi(U)))^*
\delta_j (\varphi(U))\right],
$$
recovering the action in \cite{RosL} and the Euler-Lagrange equation
\[
\varphi(U)^*\Delta(\varphi(U)) +
(\delta_1(\varphi(U)))^*\delta_1(\varphi(U)) + 
(\delta_2(\varphi(U)))^*\delta_2(\varphi(U)) = 0
\]
there.

The final example is the one treated in this paper. For any (smooth)
homomorphism $\varphi\colon A_\Theta \longrightarrow A_\theta$ and
any $(2, \infty)$-summable spectral triple $(A_\theta^\infty, \cH, D)$,
and any positive element $G \in \Omega^2(\cA_\Theta)$ (or metric on $ 
\cA_\Theta$) given by
$$
G = \sum_{j,k=1}^2 G_{ij}(dU_j)^*dU_k,
$$
the energy of $\varphi$ is
$$
\cL_{G,D}(\varphi) = \varphi^*(\psi_2)(G)  =
\sum_{j,k=1}^2 G_{jk}\Tr((1+\gamma) [D, \varphi(U_j)^*][D,  
\varphi(U_k)]) \ge 0.
$$
where $U$, $V$ are the canonical generators of $A_\Theta$.

Choosing the particular spectral triple on the noncommutative torus as
above, we
obtain the action  and Euler-Lagrange equation considered in \S3.

One can consider other choices of spectral triples on $A_\theta$
defined as follows. For instance, let
$g = \left(\begin{array}{cc} g_{11} & g_{12} \\g_{21} &
g_{22}\end{array}\right)
\in M_2(\bbR)$
be a symmetric real-valued positive definite matrix.
Then one can consider the $2$-dimensional complexified Clifford algebra,
with self-adjoint generators $\gamma_\mu \in M_2(\bbC)$ and relations
$$
\gamma_\mu \gamma_\nu + \gamma_\nu\gamma_\mu = g^{\mu\nu}, \qquad \mu,
\nu =1,2,
$$
where $(g^{\mu\nu})$ denotes the matrix $g^{-1}$.
Then with $\cH$ as before, define $D = {\displaystyle\sum_{\mu=1}^2} \gamma_\mu
\delta_\mu$. The energy in this more general case is
\begin{equation}
\label{eqn:energy-general}
\cL_{G,D}(\varphi) = \varphi^*(\psi_2)(G)  =
\sum_{j,k=1}^2 \sum_{\mu,\nu=1}^2 
G_{jk} g^{\mu\nu} \Tr(\delta_\mu(\varphi(U_j))^*\delta_ 
\nu(\varphi(U_k)) \ge 0.
\end{equation}
In this case, the trace $\Tr$ is either the Dixmier trace composed with
$D^{-2}$, or the canonical trace on $A_\theta$ multiplied by the  
factor $\sqrt{\det(g)}$, to make the
energy scale invariant. The Euler-Lagrange equations in this case are 
an easy modification of those in Proposition \ref{prop:ELeq}.

\subsection{The Wess-Zumino term} 
\label{sec:WZ}
There is a rather large literature on ``noncommutative Wess-Zumino
theory'' or ``noncommutative WZW theory,'' referred to in
\cite[\S5]{DKL0} and summarized in part in the survey articles
\cite{DN} and \cite{Sz03}. Most of this literature seems to deal with the
Wess-Zumino-Witten model (where spacetime is a compact group) or with
the Moyal product, but we have been unable to find anything that
applies to our situation where both spacetime and the worldsheet are
represented by noncommutative $C^*$-algebras (or dense subalgebras
thereof). For that reason, we will attempt here to reformulate the
theory from scratch.

The classical Wess-Zumino term is associated to a closed $3$-form $H$
with integral periods on $X$ (the spacetime manifold). If $\Sigma^2$
is the boundary of a $3$-manifold $W^3$, and if $\varphi\co\Sigma \to
X$ extends to $\widetilde\varphi\co W\to X$, the Wess-Zumino term is 
\[
\cL_{WZ}(\varphi)=\int_W (\widetilde\varphi)^*(H).
\]
The fact that $H$ has integral periods guarantees that 
$e^{2\pi i \cL_{WZ}(\varphi)}$    is
well-defined, i.e., independent of the choice of $W$ and the extension
$\widetilde\varphi$ of $\varphi$.

To generalize this to the noncommutative world, we need to dualize all
spaces and maps. We replace $X$ by $\cB$ (which in the classical case
would be $C_0(X)$), $\Sigma$ by $\cA$, and $W$ by $\cC$. Since $H$
classically was a cochain on $X$ (for de Rham cohomology), it becomes
an odd \emph{cyclic cycle} on $\cB$. The integral period condition can
be replaced by requiring
\begin{equation}
\label{eq:intper}
\langle H, u\rangle \in \bZ
\end{equation}
for all classes $u\in K^1(\cB)$ (dual $K$-theory, defined via spectral
triples or
some similar theory). The inclusion $\Sigma\hookrightarrow W$ dualizes
to a map $q\co \cC\to \cA$, and we suppose $\varphi\co \cB\to \cA$ has
a factorization
\[
\xymatrix{& \cC \ar[d]^q
\\ \cB \ar[r]^\varphi \ar@{.>}[ru]^{\widetilde\varphi}
& \,\cA.
}
\]
The noncommutative Wess-Zumino term then becomes
\[
\cL_{WZ}(\varphi)=\langle {\widetilde\varphi}_*(H), [\cC] \rangle,
\]
with $[\cC]$ a cyclic cochain corresponding to
integration over $W$. The integral period condition is relevant for
the same reason as in the classical case---if we have another
``boundary'' map $q'\co \cC'\to \cA$ and corresponding
$\widetilde\varphi'\co \cB \to \cC'$, and if
$\cC \oplus_\cA \cC'$ is
``closed,'' so that $[\cC]-[\cC']$ corresponds to a class $u\in
K^1(\cC \oplus_\cA \cC')$, then 
\[
\langle {\widetilde\varphi}_*(H), [\cC]\rangle -\langle
	{\widetilde\varphi'}_*(H), [\cC']\rangle
= \langle H, (\widetilde\varphi\oplus \widetilde\varphi')^*(u)
\rangle \in \bZ,
\]
and thus $e^{2\pi i\cL_{WZ}(\varphi)}$ is the same whether computed via $[\cC]$
or via $[\cC']$.

Now we want to apply this theory when $\cA= A_\theta$ (or a suitable
smooth subalgebra, say $A_\theta^\infty$).
If we realize $A_\theta$ as the crossed product $C^\infty(S^1)\rtimes_
\theta \bbZ$, we can view $A_\theta^\infty$ as the ``boundary''
of $\cC = C^\infty(D^2)\rtimes_\theta \bbZ$, where $D^2$ denotes the unit disk
in $\bbC$. The natural element $[\cC]$ is the trace on $\cC$ coming
from normalized Lebesgue measure on $D^2$.

To summarize, it is possible to  enhance the sigma-model action on a
spacetime algebra $\cB$ with the addition of
a Wess-Zumino term $\cL_{WZ}(\varphi)$, depending on a choice of a
``flux'' $H$.

\subsection{More general spacetimes}
\label{sec:bundles}
In references such as \cite{LNR, MR, MR1}, T-duality considerations
suggested that very often one should consider spacetimes which are not
just noncommutative tori, but ``bundles'' of noncommutative tori over
some base space, such as the {\Ca} of the discrete Heisenberg group,
called the ``rotation algebra'' in \cite{AnP}. A theory of some of
these bundles was developed in \cite{ENOO}.

For present purposes, the following definition will suffice:
\begin{definition}
\label{def:ATheta}
Let $Z$ be a compact space and let $\Theta\co Z\to\bbT$
be a continuous function from $Z$ to the circle group.
We define the \emph{noncommutative torus bundle algebra}
associated to $(Z, \Theta)$ to be the universal {\Ca}
$A=A(Z,\Theta)$ generated over a central copy of $C(Z)$
(continuous functions vanishing on the base space, $Z$) by
two unitaries $u$ and $v$, which can be thought of as continuous
functions from $Z$ to the unitaries on a fixed Hilbert space
$\cH$, satisfying the commutation rule
\begin{equation}
\label{eq:commrule}
u(z)v(z) = \Theta(z)v(z)u(z).
\end{equation}
Note that $A$ is the algebra
$\Gamma(Z,\cE)$ of sections of a continuous field $\cE$ of
rotation algebras, with fiber $A_{\log\Theta(z)/(2\pi i)}$ over
$z\in Z$.
\end{definition}
\begin{examples}
\label{ex:ATheta}
The reader should keep in mind three key examples of Definition
\ref{def:ATheta}. If $Z=\{z\}$ is a point, $A(Z, \Theta)$ is
just the rotation algebra $A_{\log\Theta(z)/(2\pi i)}$. More
generally, if $\Theta$ is a constant function with constant
value $e^{2\pi i \theta}$, then $A(Z, \Theta) = C(Z)\otimes A_\theta$.
And finally, there is a key example with a nontrivial function
$\Theta$, that already came up in \cite{MR} from T-dualization of
$\bbT^3$ (viewed as a principal $\bbT^2$-bundle
over $\bbT$)  with a nontrivial H-flux, namely the group {\Ca}
of the integral Heisenberg group. In this example, $Z=S^1=\bbT$
and $\Theta\co \bbT \to \bbT$ is the identity map. If $w$ is
the canonical unitary generator of $C(Z)$, then in this case the
commutation rule \eqref{eq:commrule} becomes simply $uv=wvu$
(with $w$ central), so as explained in \cite{AnP}, $A$ is the
universal {\Ca} on three unitaries $u$, $v$, $w$, satisfying this
commutation rule.
\end{examples}
\begin{remark}
\label{rem:AThetatoAtheta}
Let $A=A(Z,\Theta)$ be as in Definition \ref{def:ATheta}, and fix
$\theta$ irrational.
Then homomorphisms $A\to A_\theta$, not assumed necessarily to be unital,
can be identified with triples consisting of the following:
\begin{enumerate}
\item a projection $p\in A_\theta$ which represents the
image of $1\in A$,
\item a unital $*$-homomorphism $\rho$ from $C(Z)$ to $pA_\theta 
p$, and
\item a unitary representation of the Heisenberg commutation relations
\eqref{eq:commrule} into the unital {\Ca} $pA_\theta p$, with
the images of $u$ and $v$ commuting with $\rho(C(Z))$. 
\end{enumerate}
\end{remark}
Even in the case discussed above with $A=C^*(u,v,w\mid uv=wvu)$
and in the special case of unital maps, the
classification of maps $\varphi\co A\to A_\theta$ is remarkably intricate.
For example, choose any $n$ mutually orthogonal self-adjoint
projections $p_1, \cdots, p_n$ in $A_\theta$ with $p_1 +\cdots+
p_n=1$. Each $p_j A_\theta p_j$ is Morita equivalent to $A_\theta$,
and is thus isomorphic to a matrix algebra
$M_{n_j}\bigl(A_{\theta_j}\bigr)$, $\theta_j \in GL(2,\bZ)\cdot
\theta$. For each $j$, there is a unital map $\varphi_j\co A\to
M_{n_j}\bigl(A_{\theta_j}\bigr)$ sending the central unitary $w$ to
$e^{2\pi i\theta_j}$. 
Then $\varphi_1\oplus \cdots \oplus \varphi_n$ is a unital
$*$-homomorphism from $A$ to $A_\theta$ sending $w$ to 
$\sum e^{2\pi i\theta_j} p_j$.
Since $n$ can be chosen arbitrarily large, one sees 
that there are quite a lot of inequivalent maps. In this particular
example, $K_1(A)$ is a free abelian group on $3$ generators, 
$u$, $v$, and an additional generator $W\in M_2(A)$
\cite[Proposition 1.4]{AnP} ($w$ does not give
an independent element since it is the commutator of $u$ and $v$).
A notion of ``energy'' for such maps $\varphi$ may be obtained by
summing the energies of the three unitaries $\varphi(u)$,
$\varphi(v)$, and $\varphi(W)$ (for the last of these, one needs to
extend $\varphi$ to matrices over $A$ in the usual way). Estimates for the
energy can again be obtained using the results and methods of \cite{HLi}.

\section{A physical model}
\label{sec:physics}
To write the partition function for the sigma-model studied in this  
paper, recall the expression for the energy from equation 
\eqref{eqn:energy-general},
$$
\cL_{G,D}(\varphi) = \varphi^*(\psi_2)(G)  =
\sqrt{\det(g)} \sum_{j,k=1}^2 \sum_{\mu,\nu=1}^2 
G_{ij} g^{\mu\nu} \Tr(\delta_ 
\mu(\varphi(U_j))^*\delta_\nu(\varphi(U_k)).
$$
It is possible to parametrize the metrics $(g_{\mu\nu})$ by a complex  
parameter $\tau$,
$$
g(\tau)= (g_{\mu\nu}(\tau)) = \left(\begin{array}{cc} 1 & \tau_1 \\  
\tau_1 &
|\tau|^2\end{array}\right)
$$
where $\tau= \tau_1 + i \tau_2 \in \bbC$ is such that $\tau_2>0$. Note  
that $g$ is invertible with inverse given by
$$
g^{-1}(\tau) = (g^{\mu\nu}(\tau)) = \tau_2^{-2} \left(\begin{array} 
{cc} |\tau|^2 & -\tau_1 \\ -\tau_1 &
1 \end{array}\right)
$$
and $\sqrt{\det(g)} = \tau_2$. The ``genus $1$'' partition function is
$$
Z(G,z) = \int_{\tau\in \bbC, \tau_2>0} \frac{d\tau\wedge d\bar\tau} 
{{\tau_2}^2} Z(G, \tau, z)
$$
where
$$
Z(G, \tau, z) = \int \cD[\varphi] e^{-z\cL_{G,\tau}(\varphi)}/ \int  
\cD[\varphi]  .
$$
is the renormalized integral. Here $\cL_{G,\tau} = \cL_{G,D}$, where  
we emphasize the dependence
of the energy on $\tau$. This integral is much too difficult to deal  
with even in the commutative case,
so we oversimplify by considering the semiclassical approximation,  
which is a sum over the
critical points. Even this turns out to be highly nontrivial, and we  
discuss it below. In the special case
when $\Theta=\theta$ and is not a quadratic irrational, then the  
semiclassical approximation
to the partition function above is
$$
Z(G, \tau, z) \approx \sum_{m\in M/\{\pm 1\}} \sum_{A}  e^{-z\cL_{G,\tau} 
(\varphi_A)},
$$
up to a normalizing factor, in the notation as explained later in this  
section. In this approximation,
$$
Z(G,z) \approx \int_{\tau\in \bbC, \tau_2>0}  \frac{d\tau\wedge d\bar\tau} 
{{\tau_2}^2}  \sum_{m\in M/\{\pm 1\}} \sum_{A}  e^{-z\cL_{G,\tau} 
(\varphi_A)}.
$$
We expect $Z(G)$ and $Z(G^{-1})$ to be related as in the classical case  
\cite{Ovrut, Buscher2}, as a manifestation of T-duality.

In the rest of this section we specialize to a (rather oversimplified)
special case based on the results of Section \ref{sec:action-tori}.  As
explained 
before, we basically take our spacetime to be a noncommutative
$2$-torus, and for simplicity, we ignore the integral over $\tau$ (the
parameter for the metric on the worldsheet) and take $\tau=i$.

As pointed out by Schwarz \cite{Sch}, changing a noncommutative
torus to a Morita equivalent noncommutative torus in many cases
amounts to an application of T-duality, and should  not change the
underlying physics.
For that reason, it is perhaps appropriate to stabilize and take
our spacetime to be represented by the algebra $A_\Theta\otimes \cK$
($\cK$ as usual denoting the algebra of compact operators), which encodes
all noncommutative tori Morita equivalent to $A_\Theta$ at once.
(Recall $A_{\Theta'}$ is Morita equivalent to $A_\Theta$ if and only
if they become isomorphic after tensoring with $\cK$, by the 
Brown-Green-Rieffel theorem \cite{BGR}.)

Since this algebra is stable, to obtain maps into the worldsheet
algebras we should take the latter to be stable also, and thus we
consider a sigma-model based on maps $\varphi\co A_\Theta\otimes \cK
\to A_\theta\otimes \cK$, where $\theta$ is allowed to vary
(but $\Theta$ remains fixed). Via the results of Section \ref{sec:morphisms},
such maps exist precisely when there is a morphism of ordered abelian
subgroups of $\bR$, from $\bbZ + \bbZ \Theta$ to $\bbZ + \bbZ \theta$,
or when there exists $c\theta + d\in \bbZ + \bbZ \theta$,
$c\theta + d >0$,  such that
$(c\theta + d)\Theta \in \bbZ + \bbZ \theta$, i.e., when there exists
$m\in M = GL(2,\bbQ)\cap M_2(\bbZ)$ (satisfying the sign condition
$c\theta + d >0$) such that $\Theta = m\cdot \theta$
or $\theta = m^{-1}\cdot \Theta$ for the action of $GL(2,\bbQ)$
on $\bbR$ by linear fractional transformations.

Given that $\Theta = m\cdot \theta$ for some $m\in M$, the matrix
$m$ determines the map $\varphi_*\co 
K_0(A_\Theta\otimes \cK) \to K_0(A_\theta\otimes \cK)$,
which turns out to be multiplication by
\begin{equation}
\label{eq:denom}
\cD\left(m ,\,\theta\right) = \left|c\theta
  + d\right| \text{ if } m=
  \begin{pmatrix}a&b\\c&d\end{pmatrix}.
\end{equation}
(Note the similarity with the factor that appears in the
transformation law for modular forms. Also note that if $\Theta =
m\cdot \theta$, then also $\Theta =
m'\cdot \theta$ for many other matrices $m'$, since one can multiply both
rows by the same positive constant factor.) 
However, $m$ does not determine the map induced by $\varphi$ on $K_1$.
A natural generalization of Conjecture \ref{conj:minenergy}
would suggest that if $\theta=\Theta$ and $m=1$, at least if $\theta$
is not a quadratic irrational (so as to exclude the Kodaka-like maps), 
then the induced map $\varphi_*$ on $K_1$ has a matrix 
\[
A=\begin{pmatrix}p&q\\r&s\end{pmatrix} 
\]
in $SL(2,\bZ)$, and  there should be (up to gauge equivalence)
a unique  energy-mini\-mi\-zing map $\varphi \co A_\Theta\otimes \cK
\to A_\theta\otimes \cK$ with energy 
\[
4\pi^2\,(p^2 + q^2 + r^2 + s^2). 
\]
Note that $p^2 + q^2 + r^2 + s^2$ is the squared Hilbert-Schmidt norm
of $A$ (i.e., the sum of the squares of the entries). We want to
generalize this to the case of other values of $m$. 

Unfortunately, the calculation in Section \ref{sec:rational} suggests
that there may not be a good formula for the energy of a harmonic map
just in terms of the induced maps on $K_0$ and $K_1$. But  a rough
approximation to the partition function might be something like
\begin{equation}
\label{eq:partition}
Z(z) \approx \sum_{m\in M/\{\pm 1\}} \sum_{A} e^{- 4\pi^2 \cD(m, \theta)
\Vert A\Vert_{HS}^2 z}\,.
\end{equation}
The formula $ 4\pi^2 \cD(m, \theta) \Vert A\Vert_{HS}^2$ for the
energy is valid not just for the automorphisms $\varphi_A$ but also
for the map $U\mapsto u^pv^q$, $V\mapsto u^rv^s$ with
\[
A = \begin{pmatrix}p&q\\r&s\end{pmatrix},\quad \det A = n
\]
from $A_{n\theta}$ to $A_\theta$, which one can check to be harmonic,
just as in Corollary \ref{cor:satisfiesEL}. The associated map on $K_0$
corresponds to the matrix
\[
m=\begin{pmatrix}n&0\\0&1\end{pmatrix}
\]
with $\cD(m,\theta)=1$.

\end{document}